\documentclass[conference]{IEEEtran}

\IEEEoverridecommandlockouts

\usepackage[cmex10]{amsmath}
\usepackage{amssymb,amsthm}
\usepackage{graphicx}
\usepackage{subfigure}
\usepackage{cite}
\usepackage{epstopdf}
\usepackage{subeqnarray}
\usepackage{diagbox}
\usepackage{booktabs}
\usepackage{multirow}
\usepackage{multicol}
\usepackage{algorithm}
\usepackage{algorithmic}
\usepackage{setspace}
\usepackage{url}
\usepackage{xcolor}
\usepackage{array}
\usepackage{footnote}
\usepackage{enumerate}
\usepackage{framed}
\usepackage{lineno}
\usepackage{threeparttable}
\usepackage{geometry}
\makesavenoteenv{tabular}

\newcommand{\mb}{\mathbf}
\newcommand{\mbb}{\mathbb}
\newcommand{\mc}{\mathcal}

\newtheorem{theorem}{\textbf{Theorem}}
\newtheorem{lemma}{\textbf{Lemma}}

\newtheorem{corollary}{\textbf{Corollary}}
\newtheorem{definition}{\textbf{Definition}}
\newtheorem{example}{\textbf{Example}}
\newtheorem{remark}{\textbf{Remark}}

\begin{document}

\title{Parity-Check Matrix Partitioning for Efficient Layered Decoding of QC-LDPC Codes}

\author
{
\IEEEauthorblockN{
Teng~Lu$^\dag$,
Xuan~He$^\dag$,~
Peng~Kang$^\ast$,~\IEEEmembership{Member,~IEEE,}
Jiongyue~Xing$^\star$,
and~Xiaohu~Tang$^\dag$,~\IEEEmembership{Senior Member,~IEEE}}\\
	\IEEEauthorblockA{${}^\dag$School of Information Science and Technology, Southwest Jiaotong University, China\\}
	\IEEEauthorblockA{${}^\ast$Science, Mathematics and Technology (SMT) Cluster, Singapore University of Technology and Design, Singapore}
\IEEEauthorblockA{$^\star$Theory Lab, Central Research Institute, 2012 Labs, Huawei Technology Co. Ltd.}
	Email: luteng@my.swjtu.edu.cn, xhe@swjtu.edu.cn, peng$\_$kang@sutd.edu.sg, xing.jiongyue@huawei.com, xhutang@swjtu.edu.cn
\thanks{Part of this work has been published in 2021 GlobeCom Workshops.}
}

\maketitle

\begin{abstract}
In this paper, we consider how to partition the parity-check matrices (PCMs) to reduce the hardware complexity and computation delay for the row layered decoding of quasi-cyclic low-density parity-check (QC-LDPC) codes.
First, we formulate the  PCM partitioning as an optimization problem, which targets to minimize the maximum column weight of each layer while maintaining a block cyclic shift property among different layers.
As a result, we derive all the feasible solutions for the problem and propose a tight lower bound $\omega_{LB}$ on the minimum possible maximum column weight to evaluate  a solution.
Second, we define a metric called layer distance to measure the data dependency between consecutive layers and further illustrate how to identify the solutions with desired layer distance from those achieving the minimum value of $\omega_{LB}=1$, which is preferred to reduce computation delay.
Next, we demonstrate that up-to-now, finding an optimal solution for the optimization problem with polynomial time complexity is unachievable.
Therefore, both enumerative and greedy partition algorithms are proposed instead.
After that, we modify the quasi-cyclic progressive edge-growth (QC-PEG) algorithm to directly construct PCMs that have a straightforward partition scheme to achieve $\omega_{LB} $ or the desired layer distance.
Simulation results showed that the constructed codes have better error correction performance and smaller average number of iterations than the underlying 5G LDPC code.


\end{abstract}



\IEEEpeerreviewmaketitle


\section{Introduction}
	Low-density parity-check (LDPC) codes \cite{Gallager62} have been widely adopted in many applications such as wireless communications  and data storage systems \cite{5gChannel1,Fang18}, for their capacity approaching performance with iterative message passing decoding \cite{Richardson01capacity}.
	To implement the decoding of LDPC codes, a schedule is required to determine the order of updating the messages.
	One well-known message-passing schedule is the flooding schedule \cite{Kschischang1998fsd}, where all the variable-to-check (V2C) messages and the check-to-variable (C2V) messages are updated simultaneously and propagated along the edges in the Tanner graph \cite{Tanner81}.
	However, it requires a large-scale parallelism, which leads to a high decoding complexity and memory requirement for hardware implementation.

	To simplify hardware implementation and reduce the storage memories, the layered schedules were proposed in \cite{Mansour2003cnLayer,Hocevar2004cnLayer,Zhang2009cnlayer,Zhang2005vnLayer,Aslam2017vnLayer}, which update the messages in a sequential order by using the latest information.
	In particular, the row layered schedules in \cite{Mansour2003cnLayer,Hocevar2004cnLayer,Zhang2009cnlayer} operate as a sequence of check node (CN) updates while the column layered schedules in \cite{Zhang2005vnLayer} and \cite{Aslam2017vnLayer} perform variable node (VN) updates sequentially.
	As shown in \cite{Hocevar2004cnLayer}, the layered schedules can achieve faster convergence speed with significantly reduced memory requirement compared to the flooding schedule.
	However, the conventional layered schedules, e.g., \cite{Mansour2003cnLayer, Hocevar2004cnLayer,Zhang2005vnLayer}, are initially developed to process only one row/column of the parity-check matrix (PCM) at each layer, which causes a reduction of decoding throughput.

	To improve the throughput, partial parallelism have been introduced in \cite{Ueng07,Studer08,Cui08,Cui09,Sun11,Hasani19,Hasani21} for  quasi-cyclic LDPC (QC-LDPC) codes \cite{Fossorier04}.
The idea is to  process more rows/columns of the PCM in parallel at each layer.
	More specifically, to decode a QC-LDPC code with partial parallelism, the rows or columns of its PCM are divided into several groups/submatrices corresponding to different layers.
	At each layer, all rows or columns are processed in parallel, and the decoding conducts layer by layer.
	As shown in \cite{Cui08,Cui09,Sun11}, the way of partitioning the PCM  can significantly affect the hardware complexity and throughput of the decoder.
In \cite{Cui08} and \cite{Cui09}, a block cyclic shift property  among different layers is exploited to make the hardware efficiently reused between consecutive layers during the decoding process.
In addition, the maximum column/row weight of each layer is further minimized to reduce both the implementation complexity and the logic delay for the computation \cite{Cui08}.
	Nevertheless, the block cyclic shift property has not been systematically investigated for partitioning PCMs, e.g., the optimal partition scheme in terms of minimizing the maximum column/row weight of each layer is not clear.
	With massive data rate requirement for the communication systems in the fifth-generation (5G) and beyond, it is challenging for the LDPC decoders to achieve ultra-high throughput with cost-effective hardware implementation.
	As one of the practical solutions, the PCM partitioning is necessary to be systematically investigated.

In this paper, we consider how to partition the PCMs to reduce the hardware complexity and computation delay for the row layered decoding of QC-LDPC codes.
The main contributions of this paper are summarized below:
\begin{itemize}
\item We formulate the PCM partitioning as an optimization problem and propose the PCM partitioning principle to minimize the maximum column weight of each layer while maintaining a block cyclic shift property among different layers.
Notably, we derive all feasible solutions (partition schemes) for the problems and propose a tight lower bound $\omega_{LB}$ for the minimum possible maximum column weight to evaluate the quality of a solution.
\item Among the solutions which achieve the minimum possible value of the lower bound $\omega_{LB}=1$, we observe that, they may have different computation delay due to the data dependency issue between consecutive layers.
With respect to that, we define a metric, called layer distance, to quantify the above issue. We further illustrate how to identify the solutions with desired layer distance from those  solutions that achieve $\omega_{LB}=1$.
\item We demonstrate our attempt to find an optimal solution for the optimization problem by connecting it to a classic graph theory problem, and show that up-to-now, there exist no algorithms to obtain an optimal solution with polynomial time complexity.
Therefore, an enumerative partition algorithm and a greedy partition algorithm are proposed as alternatives.
\item For some cases, it may be too time-consuming to find a solution or there are no solutions achieving $\omega_{LB}$ or the desired layer distance.
Thus, we modify the quasi-cyclic progressive edge-growth (QC-PEG) algorithm \cite{Li04QCLDPC,He18PEG} to directly construct PCMs that have a straightforward solution to achieve $\omega_{LB}$ or the desired layer distance.
\item
We evaluate the performance of the proposed enumerative and greedy algorithms for partitioning the 5G LDPC codes.
There exist cases that $\omega_{LB}$ or the desired layer distance is not achieved.
Then, under the same code parameters as a 5G LDPC code, we use the modified QC-PEG algorithm to construct two QC-LDPC codes to achieve $\omega_{LB}$ and desired layer distance, respectively.
Simulation results show that the constructed codes have better error correction performance and achieve smaller average number of iterations than the 5G LDPC code.
\end{itemize}

	The remainder of this paper is organized as follows.
	Section II introduces the preliminaries of QC-LDPC codes and the row layered schedule based on the sum-product algorithm (SPA) \cite{Mansour2003cnLayer}.
	Section III formulates the PCM partitioning problem and then Section IV illustrates the characterization of it.
	%
	%
    Section V develops two algorithms to solve the PCM partitioning problem.
    The proposed greedy and enumerative algorithms for solving the PCM partitioning problem are presented in Section VI.
    Section VI develops a modified QC-PEG algorithm to design QC-LDPC codes that have a straightforward partition scheme to achieve $\omega_{LB}$ or a desired layer distance.
    Section VII investigates the partition schemes obtained by the proposed greedy and enumerative algorithms for the 5G LDPC codes and also presents the error correction performance and convergence behavior of the QC-LDPC codes constructed by the modified QC-PEG algorithm.
Finally, Section VIII concludes our work.

\emph{Notations}: We use uppercase and lowercase letters to represent global and local variables, respectively, e.g., $M$ and $m$.
Use bold face  uppercase and lowercase letters to respectively represent matrices and vectors, e.g., $\mb{H}$ and $\mb{y}$.
Use calligraphic letters to  represent sets, e.g., $\mc{G}$.
Use Greek letters to represent functions, e.g., $\phi$.
For nonnegative integers $m$ and $n$,
use $[m,n)$ to represent the integer set $\{m, m+1, \ldots, n-1\}$.
Moreover, if $m=0$, we simplify it to $[n)$.
Use $\mbb Z$ to represent the set of all integers.
For positive integers $a$ and $b$, use $a|b$ to represent that $b$ is divisible by $a$.

\section{Preliminaries}

\subsection{QC-LDPC Codes}

Let $\mb x=(x_0,\ldots, x_{z-1})$ be a vector of length $z$.
For any integer $s$, the $s$-cyclic (right) shift of $\mb x$  is defined as the vector
$$\lambda^s(\mb x)=(x_{(-s) \bmod z}, x_{(1-s) \bmod z}, ..., x_{(z-1-s) \bmod z}).$$
Assume that $\mb C$ is a matrix containing $z$ rows.
For $i \in [z)$, denote $\mb C[i]$ as the $i$-th row of $\mb C$.
Then $\mb C$ is said to be a circulant if $\mb C$ is a $z \times z$ square matrix and $\mb C[i] = \lambda^i(\mb C[0])$, $\forall i \in [z)$.

A QC-LDPC code belongs to the class of structured LDPC codes, with parity-check matrix consisting of $M \times N$ circulants as follows:
    \begin{eqnarray}
	\mathbf{H}=\left[\begin{array}{ccc}\label{Eqn_H_Matrix}
	\mb H_{0,0} & \cdots& \mb H_{0, N-1} \\
	\vdots  & \ddots & \vdots \\
	\mb H_{M-1, 0} & \cdots & \mb H_{M-1, N-1}
	\end{array}\right],
	\end{eqnarray}
where $\mb H_{m,n}$  is a  $Z \times Z$  circulant for $m \in [M)$ and  $n\in [N)$.
 As a result, $\mb H$ can be equivalently described by a corresponding base matrix $\mb{B}$:
    \begin{eqnarray}\label{Eqn_B_Matrix}
	\mathbf{B}=\left[\begin{array}{cccc}
	 \mb b_{0,0}  & \cdots &  \mb b_{0, N-1} \\
	\vdots &  \ddots & \vdots \\
	 \mb b_{M-1, 0} & \cdots &  \mb b_{M-1, N-1}
	\end{array}\right].
	\end{eqnarray}
In \eqref{Eqn_B_Matrix}, for $m \in [M)$ and  $n\in [N)$, $\mb b_{m,n}$ is an integer vector: if  $\mb b_{m,n} = (-1)$, it indicates that $\mb H_{m,n}$ is a zero matrix;
otherwise, the elements in $\mb{b}_{m,n}$
form a subset of $[Z)$ and specify the positions of the non-zero entries in the first row of $\mb H_{m,n}$.
Particularly, when $\mb b_{m,n}$ only contains a single integer, it corresponds to a zero matrix or a circulant permutation matrix.
 We hereafter remove the parentheses for simplicity, e.g., $\mb b_{m,n}=-1$.

For example, consider $Z = 4$ and the base matrix
 \begin{equation}
 \mathbf{B}=\left[\begin{array}{lll}
 1 & 3 & -1 \\
 0 & 2 & 0
 \end{array}\right].
 \end{equation}
The PCM $\mb{H}$ corresponding to $\mb{B}$ is given below:
 \begin{align}\label{eqn: PCM}
 \setstretch{0.95}
 \mathbf{H} =
 \left[\begin{array}{llll|llll|llll}
 0 & 1 & 0 & 0 & 0 & 0 & 0 & 1 & 0 & 0 & 0 & 0 \\
 0 & 0 & 1 & 0 & 1 & 0 & 0 & 0 & 0 & 0 & 0 & 0 \\
 0 & 0 & 0 & 1 & 0 & 1 & 0 & 0 & 0 & 0 & 0 & 0 \\
 1 & 0 & 0 & 0 & 0 & 0 & 1 & 0 & 0 & 0 & 0 & 0 \\\hline
 1 & 0 & 0 & 0 & 0 & 0 & 1 & 0 & 1 & 0 & 0 & 0 \\
 0 & 1 & 0 & 0 & 0 & 0 & 0 & 1 & 0 & 1 & 0 & 0 \\
 0 & 0 & 1 & 0 & 1 & 0 & 0 & 0 & 0 & 0 & 1 & 0 \\
 0 & 0 & 0 & 1 & 0 & 1 & 0 & 0 & 0 & 0 & 0 & 1
 \end{array}\right].
 \end{align}

Considering the Tanner graph \cite{Tanner81} representation for $\mb H$,
we denote $${\mc G=(\mc V,\mc E)=(\mc V_v \cup \mc V_c,\mc E)}$$ as the Tanner Graph of the LDPC code, where $\mc V_c = \{c_m: m\in [MZ)\}$ denotes the CN set, $\mc V_v = \{v_n: n \in [NZ)\}$ denotes the VN set and $\mc E \subseteq \mc V_c \times \mc V_v$ denotes the edge set.

In the rest of this paper, we consider the problem of partitioning the $\mb{H}$ given in (1).
In the following, we further introduce some concepts related to $\mb H$.

\begin{itemize}
\item \textbf{Row matrix}:  For an ordered subset $\mc A \subseteq [MZ)$, define $\mb{H}_{\mc{A}}$ as the matrix  orderly formed by the rows of $\mb H$, with indices in $\mc A$.

Taking the $\mb{H}$ in \eqref{eqn: PCM} as an example, we have
\begin{equation}\label{Eqn_H68}
\mathbf{H}_{\{5,7\}}=
\left[\begin{array}{llll|llll|llll}
\setstretch{0.6}
0 & 1 & 0 & 0 & 0 & 0 & 0 & 1 & 0 & 1 & 0 & 0 \\
0 & 0 & 0 & 1 & 0 & 1 & 0 & 0 & 0 & 0 & 0 & 1
\end{array}\right].
\end{equation}

\item \textbf{Block cyclic shift}:
For $m \in [M), i\in [Z)$ and an arbitrary integer $s$, the block $s$-cyclic shift of $\mb H_{m}[i]\triangleq(\mb H_{m,0}[i],\dots,\mb H_{m,N-1}[i])$ is defined as the vector  $\phi^s(\mb H_{m}[i])=(\lambda^s(\mb H_{m,0}[i]),\dots,\lambda^s(\mb H_{m,N-1}[i]))$.
Moreover, for a matrix $\mb H_{\mc A}$, $\phi^s(\mb H_{\mc A})$ is to apply the block cyclic shift operation on each row of $\mb H_{\mc A}$.

It is easy to see that $\mb H$ satisfies the block cyclic shift property:
For each $m \in [M)$, $\mb H_{m}[i]=\phi^{i}(\mb H_{m}[0])$ for any $i \in [Z)$.

\item \textbf{Maximum column weight}: For $i \in [NZ)$, denote $\omega_i(\mb H)$ as the Hamming weight of the $i$-th column  of $\mb H$.
Denote $\omega(\mathbf{H})$ as the maximum column weight of matrix $\mathbf{H}$, i.e., $\omega(\mathbf{H})= \max\limits_{i \in [NZ)}\omega_i(\mb H)$.
For example,  $\omega_0(\mb{H}_{\{5,7\}}) = 0$  and $\omega(\mathbf{H}_{\{5,7\}}) = 1$ in \eqref{Eqn_H68}.
\end{itemize}

\subsection{Layered Decoding Algorithm}\label{section: layered decoding}

In this paper, we concentrate on partitioning the rows of $\mb H$ for the row layered decoding.
However, our partition methods can be  applied to the column layered decoding over $\mb H$ as well.

The rows of $\mb H$ are partitioned into $L$ groups/submatrices,
where the sets of the row indices in each group
are respectively denoted by $\mc{T}_0,\dots, \mc{T}_{L-1}$,
and satisfy  the following constraint:

 \textbf{Constraint 1:} $\mc{T}_0,\dots, \mc{T}_{L-1}$ satisfies:
\begin{eqnarray*}
  &&\mc{T}_0\cup\cdots\cup \mc{T}_{L-1}=[MZ), \\
  &&\mc{T}_i \cap \mc{T}_j=\emptyset, \forall i,j \in [L), i \neq j.
\end{eqnarray*}
For each $l \in [L)$,  $\mb H_{\mc{T}_l}=[\mb H_{\mc{T}_l,0}  \cdots \mb H_{\mc{T}_l, N-1}]$ is  called the $l$-th layer.
The decoder operates layer after layer by  processing all rows of each layer  in parallel based on certain message passing algorithm.

		\begin{table}[t!]
		\begin{algorithm}[H]
			\caption{Row Layered Sum-Product Algorithm}
			\label{algo: Layer}
			\begin{algorithmic}[1]
				\REQUIRE Received LLR values $\mathbf{r} = (r_0, r_1,\ldots, r_{NZ-1})$.
				\ENSURE Hard decision estimates $\hat{\mb u}=(\hat{u}_{0}, \hat{u}_{1},\ldots, \hat{u}_{NZ-1})$.
				\STATE{\textbf{Initialization:}}
				\STATE{$\delta_{m,n}^{0} = 0$,$\forall m \in [MZ), n \in \mathcal{N}(c_m)$}.
				\STATE{$\Lambda_n = r_n$, $\forall n \in [NZ)$}.
				\FOR {$e=1$ to $I_\text{max}$}
				\FOR {$l=0$ to $L-1$}

				\STATE {Compute $\mu_{n,m} = \Lambda_n - \delta_{m,n}^{e-1} $ for each $m\in \mc T_l$ and each $n \in \mathcal{N}(c_m)$ in parallel}.
				%
				
				%
				%
				\STATE {Compute
$$
\delta_{m,n}^{e}=2 \tanh ^{-1}\left(\prod_{{n'} \in \mathcal{N}(c_m)\backslash n} \tanh \left(\frac{\mu _{n',m}}{2}\right)\right)
$$
for each $m\in \mc T_l$ and each $n \in \mathcal{N}(c_m)$ in parallel}.
				\STATE {Update $\Lambda_n = \Lambda_n + \delta_{m,n}^{e}-\delta_{m,n}^{e-1}$ for all $m\in \mc T_l$ and each $n \in \mathcal{N}(c_m)$}.
				%
				\ENDFOR
				\STATE{${{\hat u}_n} = \left\{ {\begin{array}{*{20}{l}}
						{0,}&{\Lambda _n \ge 0,}\\
						{1,}&{\text{otherwise.}}
						\end{array}} \right.$}
                \STATE{If $\hat {\mb u}$ is a codeword, stop decoding}.
				\ENDFOR

			\end{algorithmic}
		\end{algorithm}
		\vspace{-1cm}
	\end{table}

Let $\mu_{n,m}$ denote the V2C message sent from VN $v_n$ to CN $c_m$, where $n\in [NZ)$ and $m\in [MZ)$.
Let $\delta_{m,n}^e$ denote the C2V message passed from CN $c_m$ to VN $v_n$ in the $e$-th iteration.
The log-likelihood ratio (LLR) of a posterior probability for VN $v_n$ is represented by $\Lambda_n$, where $n\in [NZ)$.
We represent the index set of CNs in the graph connected to $v_n$ by $\mathcal{N}(v_n)$ and the index set of VNs connected to $c_m$ by $\mathcal{N}(c_m)$.
Denote the index set of all VNs connected to CN $c_m$ except $v_n$ by $\mathcal{N}(c_m)\backslash n$.
Assume that a codeword $\mathbf{u} = (u_0, u_1,\ldots,u_{NZ-1})$ is transmitted and $\mathbf{r} = (r_0,r_1, \ldots,r_{NZ-1})$ is the received LLR values from the channel output with $r_n$ for node $v_n$.

The row layered decoding based on the sum-product algorithm (SPA) \cite{Mansour2003cnLayer} is summarized in Algorithm 1.
In practice, it is very desirable that the hardware can support both the decoding parallelism within a layer and reusability between consecutive layers for achieving high throughput.
In addition, the power efficiency of the decoder is also of great concern for applications with limited power supply such as the Internet of things.
We discuss the details as follows.

\textbf{Parallelism:}
As shown in Lines 6-8 of Algorithm 1, the decoder processes all rows in
parallel within a layer.
According to Line 8, for different $n \in [NZ)$, the computations of $\Lambda_n$s are independent.
Thus, the computations of $\Lambda_n$s can be carried out in parallel for each VN.
On the other hand, for each $n \in [N)$, the computation of $\Lambda_n$ is carried out for $\omega_n(\mb H_{\mc T_l})$ times, which are desired to be carried out in serial.
Therefore, for each layer, the computing delay depends on the VN with the largest column weight.
To reduce the implementation complexity and logical delay of the decoding,  it is necessary to minimize the maximum column weight for each layer.

\textbf{Reusability:} 
To simplify the hardware implementation, reusing the hardware architecture of the first layer is preferred in practice.
In \cite{Cui08,Cui09}, a reuse method between consecutive layers are summarized below.

\begin{definition}[\cite{Cui08,Cui09}]
Let $\mb H$ be the PCM  in (1).
The $L$ layers $\mb H_{\mc T_0}, \mb H_{\mc T_1}, ..., \mb H_{\mc T_{L-1}}$ are said to  have the  block $S$-cyclic shift property if $\mb H_{\mc{T}_l}=\phi^{lS}(\mb H_{\mc{T}_0})$
for all $l \in [L)$.
\end{definition}

Note that if the $L$ layers have the block cyclic shift property, only setting of 0-th layer $\mb H_{\mc T_0}$ requires to be configured specifically for the hardware architecture,
while the other layers can reuse the hardware by  a fixed shift operator transforming it into the pervious one.
Consequently, the hardware reuse can be greatly simplified, compared to a partition scheme without the block cyclic shift property \cite{Cui08,Cui09}.

\textbf{Power efficiency:}
The power consumption of decoding a layer is proportional to the number of its non-zero entries.
Accordingly, the peak (resp. average) power consumption corresponds to the largest (resp. average) number of non-zero entries of each layer.
Thus, the PCM partition scheme with block cyclic shift property can achieve  high power efficiency because its peak-to-average power ratio can be minimized to 1 for the decoder.
\section{Problem Formulation of PCM Partitioning}\label{section: principle}

\subsection{Basic Optimization Problem}

Let $\mb{H}$ be the PCM of a QC-LDPC code with size $MZ \times NZ$  given in (1).
Suppose that $\mb{H}$ satisfies the following two properties:
\begin{itemize}
\item There are no zero rows  otherwise we can remove them; and
\item $\mb{H}$ does not contain two identical rows otherwise there are a large amount of 4-cycles greatly reducing performance.
\end{itemize}

Given integer $L>1$, our goal is to partition the rows of $\mb{H}$ into $L$ layers such that Constraint 1 is satisfied and
the $L$ layers have
\begin{itemize}
\item  block cyclic shift property; and
\item  minimum maximum column weight.
\end{itemize}
This can benefit the row layered decoding by reducing the decoding latency within a layer and  simplifying the hardware reusability between consecutive layers.
Specifically, we formulate the partitioning problem as the following optimization problem.

 \textbf{Problem 1:}   $\min\limits_{S > 0, \,\mc{T}_0,\mc{T}_1,\dots,\mc{T}_{L-1}} \,\,\omega (\mb{H}_{\mc{T}_0})$
 \begin{IEEEeqnarray}{l}
 \,s.t. \quad \text{Constraint 1 is satisfied, and}\notag\\
 \quad \mb{H}_{\mc{T}_l}=\phi^{lS}( \mb{H}_{\mc{T}_0}), \forall l \in [L).\label{optimize:condition3}
\end{IEEEeqnarray}
Constraint 1 guarantees that  the $L$ layers form a partition of $\mb H$,
and \eqref{optimize:condition3} ensures its block $S$-cyclic shift property,

Note from \eqref{optimize:condition3} that  $\mc T_1,\dots, \mc T_L$ are determined by $S$ and $\mc T_0$,
and all the layers have the same maximum column weight.
Therefore,  we can use $(S,\mc T_0)$ to represent a partition scheme uniquely and
just focus on the maximum column weight of $\mb {H}_{\mc{T}_0}$, i.e. $\omega (\mb {H}_{\mc{T}_0})$.
Any $(S,\mc T_0)$ satisfying Constraint 1 and \eqref{optimize:condition3} is called a feasible solution to Problem 1.
A feasible solution $(S^*,\mc T^*_0)$ is called the optimal solution if $w(\mb H_{\mc T^*_0}) = \min\limits_{S > 0, \,\mc{T}_0,\mc{T}_1,\dots,\mc{T}_{L-1}} \,\,\omega (\mb{H}_{\mc{T}_0})$.

\subsection{Data Dependency}

The optimal solution to Problem 1 provides an efficient PCM partition scheme by configuring the block cyclic shift value $S$ and the set of row indices $\mc T_0$, which can well support the parallelism, reusability, and power efficiency of a layered LDPC decoder.
However, it is possible to further optimize the PCM partition scheme from data dependency aspect for achieving shorter computation delay \cite{Zhang2009cnlayer}.
Data dependency occurs in the message update process when two CNs
have the same neighbour VN.
The following is an illustrative example.

\begin{example}\label{Examp_1}
Consider $S=1$, $L=4$ and $\mb H$ in \eqref{eqn: PCM}.
Note that $\omega (\mb{H}_{\mc{T}_0})\ge 1$, otherwise $\mb H_{\mc T_0}$ will be a zero matrix.
It is easy to verify that $(S,\mc T_0)=(1,\{0,4\})$ and $(S',\mc T'_0)=(1,\{0,7\})$ are two optimal solutions to Problem 1.
whose  partitioned matrices  are as follows.

 \begin{align*}
 \setstretch{0.9}
 \setlength{\arraycolsep}{4pt}
 \mathbf{H}_{\mc T_0} =\mathbf{H}_{\{0,4\}}=
 \left[\begin{array}{llll|llll|llll}
 0 & 1 & 0 & 0 & 0 & 0 & 0 & 1 & 0 & 0 & 0 & 0 \\
 1 & 0 & 0 & 0 & 0 & 0 & 1 & 0 & 1 & 0 & 0 & 0
 \end{array}\right],
 \end{align*}

  \begin{align*}
 \setstretch{0.9}
 \setlength{\arraycolsep}{4pt}
 \mathbf{H}_{\mc T_1} =\mathbf{H}_{\{1,5\}}=
 \left[\begin{array}{llll|llll|llll}
  0 & 0 & 1 & 0 & 1 & 0 & 0 & 0 & 0 & 0 & 0 & 0 \\
 0 & 1 & 0 & 0 & 0 & 0 & 0 & 1 & 0 & 1 & 0 & 0 \\
 \end{array}\right],
 \end{align*}

   \begin{align*}
 \setstretch{0.9}
 \setlength{\arraycolsep}{4pt}
 \mathbf{H}_{\mc T_2} =\mathbf{H}_{\{2,6\}}=
 \left[\begin{array}{llll|llll|llll}
  0 & 0 & 0 & 1 & 0 & 1 & 0 & 0 & 0 & 0 & 0 & 0 \\
 0 & 0 & 1 & 0 & 1 & 0 & 0 & 0 & 0 & 0 & 1 & 0 \\
 \end{array}\right],
 \end{align*}

   \begin{align*}
 \setstretch{0.9}
 \setlength{\arraycolsep}{4pt}
 \mathbf{H}_{\mc T_3} =\mathbf{H}_{\{3,7\}}=
 \left[\begin{array}{llll|llll|llll}
  1 & 0 & 0 & 0 & 0 & 0 & 1 & 0 & 0 & 0 & 0 & 0 \\
 0 & 0 & 0 & 1 & 0 & 1 & 0 & 0 & 0 & 0 & 0 & 1 \\
 \end{array}\right].
 \end{align*}

\begin{align*}
 \setstretch{0.9}
 \setlength{\arraycolsep}{4pt}
 \mathbf{H}_{\mc T'_0} =\mathbf{H}_{\{0,7\}}=
 \left[\begin{array}{llll|llll|llll}
 0 & 1 & 0 & 0 & 0 & 0 & 0 & 1 & 0 & 0 & 0 & 0 \\
  0 & 0 & 0 & 1 & 0 & 1 & 0 & 0 & 0 & 0 & 0 & 1
 \end{array}\right],
 \end{align*}

  \begin{align*}
 \setstretch{0.9}
 \setlength{\arraycolsep}{4pt}
 \mathbf{H}_{\mc T'_1} =\mathbf{H}_{\{1,4\}}=
 \left[\begin{array}{llll|llll|llll}
  0 & 0 & 1 & 0 & 1 & 0 & 0 & 0 & 0 & 0 & 0 & 0 \\
 1 & 0 & 0 & 0 & 0 & 0 & 1 & 0 & 1 & 0 & 0 & 0 \\
 \end{array}\right],
 \end{align*}

   \begin{align*}
 \setstretch{0.9}
 \setlength{\arraycolsep}{4pt}
 \mathbf{H}_{\mc T'_2} =\mathbf{H}_{\{2,5\}}=
 \left[\begin{array}{llll|llll|llll}
  0 & 0 & 0 & 1 & 0 & 1 & 0 & 0 & 0 & 0 & 0 & 0 \\
0 & 1 & 0 & 0 & 0 & 0 & 0 & 1 & 0 & 1 & 0 & 0 \\
 \end{array}\right],
 \end{align*}

   \begin{align*}
 \setstretch{0.9}
 \setlength{\arraycolsep}{4pt}
 \mathbf{H}_{\mc T'_3} =\mathbf{H}_{\{3,6\}}=
 \left[\begin{array}{llll|llll|llll}
  1 & 0 & 0 & 0 & 0 & 0 & 1 & 0 & 0 & 0 & 0 & 0 \\
 0 & 0 & 1 & 0 & 1 & 0 & 0 & 0 & 0 & 0 & 1 & 0  \\
 \end{array}\right].
 \end{align*}
\end{example}

In partition scheme 1, the data dependency occurs during the CN update of the nodes $c_0$ and $c_5$  located in layers 0 and 1.
At layer 1,  the $v_1$-to-$c_5$ message needs to be computed (Line 6 in Algorithm 1) based on the  posterior LLR $\Lambda_1$
updated  at  layer 0 (Line 8 in Algorithm 1).
However, due to the delay caused by the pipelining stage \cite{Cui08,Cui09},
one has to pay extra waiting time to get the latest value of $\Lambda_1$.

Whereas partition scheme 2 can avoid data dependency between two consecutive layers,
where all the $\Lambda_n$s ($n \in [12)$) are updated  at most once in any two consecutive layers.
 As a consequence, the computation delay can be reduced  for obtaining the latest  posterior LLRs.

That is,  the second solution leads to a shorter computation delay compared to first one, though they have the same maximum column weight.
In summary, to reduce the  computation delay, the column weight  for the consecutive layers  should be as small as possible.

\subsection{Further Optimization Problem}

In practice, there is often more than one optimal solution to Problem 1, where the data dependency of each solution may vary,
as illustrated in Example \ref{Examp_1}.
To minimize the computation delay, we intend to find out the optimal solution that can decrease the data dependency among  consecutive layers.
The general idea is to separate any two non-zero entries in a column as far as possible, with respect to layer distance.

\begin{definition}[Layer distance]\label{def: layer distance}
Given $\mb H$ and $L$, let $(S,\mc T_0)$ be a feasible solution to Problem 1.
The layer distance of $(S, \mc T_0)$, denoted by $d(S, \mc T_0)$, is defined as the maximum integer $l \in [L)$ such that
   \begin{align}\label{Eqn_LD}
 \omega \left(\left[\begin{array}{c}
  \mathbf{H}_{\mc T_0} \\
\mathbf{H}_{\mc T_1} \\
\vdots \\
\mathbf{H}_{\mc T_{l-1}}
 \end{array}\right]\right)\leq 1.
 \end{align}
\end{definition}

\begin{remark}
1) In general, the layer distance can be defined as the minimum maximum integer $l \in [L)$ such that \eqref{Eqn_LD} still holds
when $\mc T_{0}, \mc T_{1},\ldots, \mc T_{l-1}$ are respectively replaced by $\mc T_{i}, \mc T_{i+1},\ldots, \mc T_{i+l-1}$ where $i$ ranges over $[L)$ and the subscript is modulo $L$. Actually, $i=0$
is enough for our definition because the resultant equation is the same as  \eqref{Eqn_LD}  by means of block cyclic shift property of $\mb H$.

2) $d(S, \mc T_0)=2$ is generally sufficient for practical applications  \cite{Zhang2009cnlayer}.

3) Note that, for $l=0$, we define $\omega(\emptyset) = 0$.
Therefore, we always have $d(S,\mc T_0) \geq 0$.
In particular, we have $d(S, \mc T_0) = 0$ if $\omega(\mb H_{\mc T_0}) > 1$. Accordingly, the data dependency is the strongest, since it occurs within a layer.
\end{remark}

According to Definition \ref{def: layer distance}, a partition scheme with layer distance $k$ can avoid data dependency within any $k$ consecutive layers.
Adding the layer distance into the objective functions, we have the following multi-objective optimization problem.

 \textbf{Problem 2:}   $\min\limits_{S > 0, \,\mc{T}_0} \,\,\omega (\mb{H}_{\mc{T}_0}),-d(S,\mc T_0)$
%
 \begin{IEEEeqnarray}{l}
 \,s.t. \quad \text{Constraint 1 is satisfied, and}\notag\\
 \quad \mb{H}_{\mc{T}_l}=\phi^{lS}( \mb{H}_{\mc{T}_0}), \forall l \in [L),\notag
\end{IEEEeqnarray}

In Problem 2, we mean to both minimize $\omega (\mb{H}_{\mc{T}_0})$ and maximize the layer distance $d(S,\mc T_0)$.
Noting that $d(S,\mc T_0)>1$ only if $\omega(\mb H_{\mc T_0})=1$.
Therefore, we should first guarantee $\omega(\mb H_{\mc T_0})=1$, and then further optimize $d(S,\mc T_0)$.
In this case,  the solution $(S, \mc T_0)$ is also one of optimal solutions to Problem 1.

\subsection{Relationship Between Problems 1 and 2}
The following theorem establishes the connection between Problem 1 and Problem 2.
\begin{theorem}\label{theorem: conflict}
Given $\mb H$, $L$, and a feasible solution $(S,\mc T_0)$,  $d(S,\mc T_0) \geq k$ if and only if (iff) $\mb H^{(S,k)}_{\mc T_0}$ is a binary matrix and $\omega( \mb H^{(S,k)}_{\mc T_0})=1$, where $\mb H^{(S,k)}\triangleq\mb H+\phi^S(\mb H)+\phi^{2S}(\mb H)+\cdots+\phi^{(k-1)S}(\mb H)$.
\end{theorem}
\begin{proof}
The necessity is obvious because $\mb H$ is a binary matrix of form \eqref{Eqn_H_Matrix}.
Since   \begin{align*}
 \omega( \mb H^{(S,k)}_{\mc T_0})
 =\omega\left(\sum\limits_{i \in [k)}\phi^{iS}(\mb H_{\mc T_0})\right)
 =\omega \left(\left[\begin{array}{c}
  \mathbf{H}_{\mc T_0} \\
\mathbf{H}_{\mc T_1} \\
\vdots \\
\mathbf{H}_{\mc T_{k-1}}
 \end{array}\right]\right)
 \end{align*}
  for binary $\mb H^{(S,k)}$, the sufficiency  then follows from Definition \ref{def: layer distance}.
\end{proof}

According to Theorem \ref{theorem: conflict},
we can get a partition scheme with desired layer distance $k$,
by solving Problem 1 for $\mb H^{(S,k)}$.
Specifically, we first construct $\mb H^{(S,k)}$,
and then find an optimal solution  $(S^*, \mc T^*_0)$ to Problem 2 for $\mb H^{(S,k)}$.
If $\omega (\mb H^{(S,k)}_{\mc T_0^*})=1$, we have $d(S^*,\mc T_0^*)\geq k$;
otherwise it implies that there does not exist  feasible solution with layer distance at least $k$.

Therefore, we mainly consider Problem 1 in what follows.

\section{Characterization of Problem 1}
\subsection{Problem Transformation}

\begin{definition}\label{def: pi}
For an integer $i \in [MZ)$ and an integer $s$, define $\pi^s(i)=Z \cdot \lfloor i/Z \rfloor + ((i+s) \mod Z)$.
Moreover, for an index set $\mc T \subseteq [MZ)$, define $\pi^{s}(\mc T)$ as $ \{\pi^{s}(x) : x \in \mc T\}$.
\end{definition}

\begin{lemma}\label{theorem:pi}
For $i, j \in [MZ)$ and an integer $s$, $\mb H[j]=\phi^s(\mb H[i])$ iff
$j=\pi^s(i)$.
\end{lemma}

\begin{proof}
According to block cyclic shift property of $\mb H$, we have $\phi^s(\mb H_{m}[k])=\mb H_{m}[(k+s) \mod Z]$.
Then, sufficiency is clear. Whereas, the necessity also follows otherwise it contradicts the assumption that
there are no two identical rows in $\mb H$.
\end{proof}

%

Then, it is  more convenient  to describe Problem 1  in terms of the row indices of $\mb H$.

 \textbf{Problem 3:}   $\min\limits_{S > 0, \,\mc{T}_0} \,\,\omega (\mb{H}_{\mc{T}_0})$
 \begin{IEEEeqnarray}{l}
  \,s.t. \quad \text{Constraint 1 is satisfied, and}\notag\\
 \quad\quad\,\,\,  \mc T_l =\pi^{lS}(\mc T_0), \forall l \in [L)\label{optimize:Tl}.
\label{optimize:condition6}
\end{IEEEeqnarray}


To solve this problem, a general approach is to firstly enumerate all the values of $(S,\mc T_0)$, next check whether it is a feasible solution, and finally obtain
an optimal solution.
It is clear that there are $MZ \choose MZ/L$ options of $\mc T_0$, i.e., the search space is prohibitively large.  Accomplishing this task  is impossible for a large size of $\mb H$.
Instead, we consider to directly characterize all the feasible solutions to Problem 3.
\subsection{Feasible Solutions to Problem 3}

To solve Problem 3,  it suffices to determine $\mc T_0$  satisfying \eqref{optimize:condition6} and Constraint 1. In this subsection,
we characterize all the feasible solutions to Problem 3.
For convenience, denote $\mc T_{l,m}= \mc T_l \cap [mZ, (m+1)Z)$ for each $m \in [M)$.
According to Lemma \ref{theorem:pi},   $\mc T_{l,m}=\pi^{lS}(\mc T_{0,m})$ and
 $\mc T_{0,m}, \ldots , \mc T_{L-1,m}$ form a partition of $[mZ, (m+1)Z)$,
which results in
\begin{itemize}
\item [\textbf{B1}] $|\mc T_{l,m}|=Z/L$,
\end{itemize}
i.e., it is further required that
\begin{itemize}
\item $L$  should be a factor $Z$.
\end{itemize}

We begin with the   $s$-cyclic shift class.
\begin{definition}[$s$-cyclic shift class]
For $i \in [MZ)$ and $s \in [1, Z)$, define the $s$-class (short for $s$-cyclic shift class) generated by $i$ as the set
$ \{i\}_s = \{\pi^{rs}(i):r \in \mbb{Z}\}$.
\end{definition}
\begin{lemma}\label{lemma:number}
Let $(S, \mc T_0)$ be a feasible solution to Problem 3. For any $l \in [L)$ and $i \in [MZ)$,  $i \in \mc T_0$ iff $\pi^{lS} (i) \in \mc T_l$.
\end{lemma}
\begin{proof}
According to Definition \ref{def: pi}, $\pi^s (\cdot)$ is a bijection.
Then, the lemma follows from \eqref{optimize:Tl}.
\end{proof}

Obviously,   $\{i\}_{S}=\{i\}_{S'}$  for all $i \in [MZ)$ with $S'=\gcd(Z, S)$. Thus, hereafter we always assume that
$S$  is a factor of $Z$, i.e., $S|Z$. 

First of all, we focus on the characterization of $\mc T_{0,0}$.  Recall that $\mc T_{0,0}, \ldots , \mc T_{L-1,0}$ form a partition of $[Z)$, each of size
$Z/L$. On the other hand, $S$-classes $\{0\}_S, \ldots , \{S-1\}_S$ form a partition of $[Z)$, each of size $Z/S$. Thus,
$\mc T_{0,0}$ must contain some  elements in  $\{i\}_S$ for some $i\in [S)$.  
Let $j\in \{i\}_S \cap \mc T_{0,0}$, i.e.,  $\{i\}_S=\{j\}_S=\{(j+lS)\bmod\,Z: l=0,\ldots, Z/S-1\}$.
For any $l \in [1,L)$, $\pi^{lS}(j) \in \mc T_{l,0}$ by \eqref{optimize:condition6}, which implies $\pi^{(L-l)S}(j) \notin \mc T_{0,0}$ since $\mc T_{0,0} \cap \mc T_{L-l,0}=\emptyset$.
Then, according to Lemma \ref{lemma:number}, we have $\pi^{LS}(j) \in \mc T_{0,0}$ since $\pi^{LS}(j)=\pi^{(L-l)S+lS}(j)= \pi^{lS}(\pi^{(L-l)S}(j))\notin \mc T_{l,0}$, for any $l \in [1,L)$.
Similarly, $\pi^{2LS}(j),\pi^{3LS}(j),\dots, \pi^{Z-LS}(j)$ are  in $\mc T_{0,0}$, which and $j,\pi^{LS}(j)$ all belong to the same $LS$-class in $\{i\}_S$.
Therefore, given $i\in [S)$,  any element $j\in \{i\}_S$ can be a candidate element for set $\mc T_{0,0}$. Moreover, once the element $j$ is selected into $\mc T_{0,0}$, the remaining elements $\pi^{S}(j),\pi^{2S}(j),\dots, \pi^{Z-S}(j)$ in $\{i\}_S$
are automatically divided into the sets $\mc T_{0,0}, \ldots , \mc T_{0,L-1}$,
as shown in Table I. It is observed from Table I that
\begin{itemize}
\item  $S|(Z/ L)$, i.e., $S$ must be a factor of $Z/L$,
\end{itemize}
otherwise  the cardinality of $\mc T_{0,l}$ contradicts B1 for some  $l\in [L)$.

\begin{table}
\centering
\caption{Elements of $\mc T_{l,0}, \forall l \in [L)$}\label{table:calss}
\begin{tabular}{|c|c|c|c|c|}
\hline
$\mc T_{0,0}$   & $\mc T_{1,0}$   & $\cdots$ &  $\mc T_{L-1,0}$\\
\hline
$j$  & $\pi^S(j)$   &$\cdots$ &  $\pi^ {(L-1)S}(j)$ \\
$\pi ^{LS}(j)$  & $\pi ^{(L+1)S}(j)$   &$\cdots$ &  $\pi ^{(2L-1)S}(j)$ \\
$\vdots$ & $\vdots$ & $\ddots$ & $\vdots$\\
$\pi ^{Z-LS}(j)$  & $\pi ^{Z-(L-1)S}(j)$   &$\cdots$ &  $\pi ^{Z-S}(j)$ \\
\hline
\end{tabular}

\end{table}

Clearly, the $L$  sets $mZ+\{i\}_S=\{mZ+((i+rS)\bmod\,Z)\},\,i=0,\ldots,S-1$, are $S$-classes forming a  partition of $[mZ, (m+1)Z)$, $m\in [M)$. So, the above selection
method can be similarly carried on $[mZ, (m+1)Z)$ to obtain $\mc T_{0,m}$.

Given $S|(Z/L)$ and  $m\in [M)$,
\begin{itemize}
\item $[mZ, (m+1)Z)$ can be partitioned into $S$ different $S$-classes:  $\mc C_{m,0}\triangleq mZ+\{0\}_S,\ldots, \mc C_{m,S-1}\triangleq mZ+\{S-1\}_S$;

\item for each $s\in [S)$, the $S$-class $\mc C_{m,s}$ can be further partitioned into $L$ different $LS$-classes: $\mc C_{m,s,0}\triangleq mZ+s+\{0\}_{LS},\ldots,$ $\mc C_{m,s,L-1}\triangleq mZ+s+\{(L-1)S\}_{LS}$.

\end{itemize}
Then, the set $\mc T_0$ consists of $MS$ different $LS$-classes, each of which is from a different $S$-class, as  characterized by the following theorem.

\begin{theorem}\label{theorem: feasible solution}
$(S,\mc T_0)$ is a feasible solution to Problem 3 iff
$$\mc T_0 \in \{\cup_{m\in [M),s \in [S)} \mc C_{m,s,\l_s}: \l_s~\mathrm{ranges ~over~}[L)\}.$$
\end{theorem}

\begin{proof}
Assume $(S,\mc T_0)$ is a feasible solution, $\mc T_0=\cup_{m\in [M)} \mc T_{0,m}= \cup_{m\in [M),s \in [S)} (\mc T_{0,m} \cap \mc C_{m,s})$.
As investigated above, $\mc T_{0,m} \cap \mc C_{m,s}$ is an $LS$-class, i.e., $\mc T_{0,m} \cap \mc C_{m,s} \in \{\mc C_{m,s,\l_s}: \l_s \in [L)\}$.
Therefore, $\mc T_0 \in \{\cup_{m\in [M),s \in [S)} \mc C_{m,s,\l_s}: \l_s \in [L)\}.$
That is, the necessity is clear.
Whereas, the sufficiency can be easily verified.
\end{proof}

\begin{example}
Let $\mathbf{H}$ be given by (\ref{eqn: PCM}) and $(L, S)=(2, 2)$.
The rows of $\mathbf{H}$ can be partitioned into four $S$-classes (2-classes) and eight $LS$-classes (4-classes) as shown in \eqref{eqn: Sclasses} and \eqref{eqn: C11-C41}, respectively.
\begin{align}
\setstretch{0.8}
\setlength{\arraycolsep}{1pt}
   &\mc C_{0,0} = \{0,2\},
  \mc C_{0,1} = {\{1,3\}} ,
  \mc C_{1,0} = {\{4,6\}} ,
  \mc C_{1,1} = {\{5,7\}} , \label{eqn: Sclasses}\\
&\begin{array}{l}
  \mc C_{0,0,0} = \{0\},
  \mc C_{0,0,1} = {\{2\}} ,
  \mc C_{0,1,0} = {\{1\}} ,
  \mc C_{0,1,1} = {\{3\}} , \\
  \mc C_{1,0,0} = {\{4\}} ,
  \mc C_{1,0,1} = {\{6\}} ,
  \mc C_{1,1,0} = {\{5\}} ,
  \mc C_{1,1,1} = {\{7\}} .
\end{array}\label{eqn: C11-C41}
\end{align}
\end{example}
According to Theorem \ref{theorem: feasible solution}, by taking an $LS$-class from each $S$-class, $(2, \mc C_{0,0,0} \cup \mc C_{0,1,0} \cup\mc C_{1,0,0} \cup\mc C_{1,1,1} )$, i.e., $(2,\{0,1,4,7\} )$ is one of the feasible solutions to Problem 3.

\subsection{Combinatorial Characterization}

In this subsection,  we  restrict  $\omega (\mb H_{\mc T_0})=1$ in the objective function
of Problem 3 to get Problem 4. Then, with the construction of all the feasible solutions in Theorem \ref{theorem: feasible solution}, we are able
to give the combinatorial characterization of Problem 4 from the perspective of of graph theory,
which shows the difficulty to solve Problem 3 by a polynomial time method.

\textbf{Problem 4:}   $ \text{Find} ~ (S, \mc T_0)$
 \begin{IEEEeqnarray*}{l}
  \,s.t. \quad \text{Constraint 1 is satisfied,}\notag\\
 \quad \quad \,\,\,\mc T_l =\pi^{lS}(\mc T_0), \forall l \in [L), \text{and}\\
 \quad \quad \,\,\,\omega (\mb H_{\mc T_0})=1.
\end{IEEEeqnarray*}

If Problem 3 is solved, then Problem 4 can be also solved,  but not vice versa.
Indeed,  Problem 4 can be straightforwardly connected to the classical  graph theory problem of finding $k$-cliques in $k$-partite graphs \cite{T2002Finding, 2013Mirghorbani}.

 A graph is denoted as $\mc{G}=(\mc{V}, \mc{E})$ with vertex set $\mc{V}$ and edge set $\mc{E}$, where
a partite $\mc P$ is a subset of $\mc{V}$ such that no vertices in $\mc P$ are connected by an edge.
Graph $\mc{G}$ is called $k$-partite graph if $\mc V$ can be partitioned into $k$ disjoint partite.
 A clique $\mc{D}$ is a subset of $\mc{V}$ such that any two distinct vertexes in $\mc{D}$ are connected by an edge in $\mc{E}$,
 which is called a $k$-clique if  $|\mc D|=k$.

We now construct a graph $\mc{G}=(\mc{V}, \mc{E})$ based on  the $LS$-classes as follows:
\begin{itemize}
\item Let $\mc{V}=\{\mc C_{m,s,l}:m \in [M), s \in [S), l \in [L)\}$.
\item Let $\mc E=\{ (\mc C_{m,s,l}, \mc C_{m',s',l'}) : m,m' \in [M), s,s' \in [S), l,l' \in [L) \}$ such that
    \begin{itemize}
    \item [\textbf{C1}] The $LS$-classes $\mc C_{m,s,l}$ and $\mc C_{m',s',l'}$ are from different $S$-classes, i.e., $m \neq m'$ or $s \neq s'$; and
    \item [\textbf{C2}] $\omega(\mb H_{\mc C_{m,s,l} \cup \mc C_{m',s',l'}})=1$.
    \end{itemize}
\end{itemize}

According to C1, for each $m \in [M)$ and $s \in [S)$, $\{\mc C_{m,s,0}, \mc C_{m,s,1},\ldots,\mc C_{m,s,L-1}\}$ forms a partite in $\mc G$,
and thus $\mc G$ is an $MS$-partite graph. Further, the following theorem illustrates that a feasible solution to Problem 4 corresponds to an $MS$-clique in $\mc G$.

\begin{theorem}
For a given $S$, let $\mc{D} = \{\mc C_{m,s,\l_s}:m \in [M), s \in [S)\}$. Then, $\mc{D} \subseteq \mc{G}$ is an $MS$-clique in $\mc{G}$ iff $(S, \mc T_0 = \cup_{\mc{C}\in \mc{D}} \mc{C})$ is a solution to Problem 4.
\end{theorem}

\begin{proof} The  proof is based a fact
    \begin{itemize}
    \item [\textbf{C3}] $\omega(\mb H_{\mc{C}})\ge 1$ for any $\mc{C}\in \mc{D}$
    \end{itemize}
since there are no zero rows in $\mb H$.

\emph{Necessity:}
If $\mc{D} $ is an $MS$-clique in $\mc{G}$,
 according to C2,  then  $\omega(\mb H_{\mc{C} \cup \mc{C}'})=1$ for any two distinct $LS$-classes $\mc{C},\mc{C}' \in \mc{D}$.
 As a result, we have $\omega(\mb H_{\cup_{\mc{C} \in \mc{D}}\, \mc{C}})=1$.
(Otherwise, we have $\omega(\mb H_{\cup_{\mc{C} \in \mc{D}} \mc{C}})\ge 2$ by C3. Consequently, there exists two $LS$-classes  $\mc{C},\mc{C}'$ with $\omega(\mb H_{\mc{C} \cup \mc{C}'})\ge 2$, leading to a contradiction.)
The necessity then follows from Theorem \ref{theorem: feasible solution}.

\emph{Sufficiency:}
Let $(S, \mc T_0 = \cup_{\mc{C}\in \mc{D}} \mc{C})$ be a solution to Problem 4.
We then have $\omega(\mb H_{\cup_{ \mc{C}\in \mc{D}}\,\mc{C}}) =1$, which and C3 indicate that any two $LS$-classes $\mc{C},\mc{C}'\in \mc{D}$ satisfy $\omega(\mb H_{\mc{C} \cup \mc{C}'})=1$.
Consequently, we have the edge $(\mc{C}, \mc{C}') \in \mc{E}$, implying that
the set $\mc{D} = \{\mc C_{m,s,\l_s}:m \in [M), s \in [S)\}$ is a clique.
\end{proof}

\begin{example} Continued with Example 2.
The corresponding $\mc{G}=(\mc{V}, \mc{E})$ is shown by Fig. \ref{fig: 2-partite}, where vertexes are denoted by circles.
As can be seen from Fig. \ref{fig: 2-partite}, $\mc{G}$ consists of 4 partites:
 $\{ \mc C_{0,0,0}, \mc C_{0,0,1}\},\{ \mc C_{0,1,0}, \mc C_{0,1,1}\},\{ \mc C_{1,0,0}, \mc C_{1,0,1}\},\{ \mc C_{1,1,0}, \mc C_{1,1,1}\}$,
 and the four dashed circles $\{ \mc C_{0,0,0}, \mc C_{0,1,0}, \mc C_{1,0,0}, \mc C_{1,1,1}\}$ form a 4-clique in $\mc{G}$.
It can also be easily verified that $\omega(\mb H_{\mc C_{0,0,0} \cup \mc C_{0,1,0} \cup\mc C_{1,0,0} \cup\mc C_{1,1,1}} )=1$, implying $(S, \mc T_0) = (2, \mc C_{0,0,0} \cup \mc C_{0,1,0} \cup\mc C_{1,0,0} \cup\mc C_{1,1,1} )$ is a solution to Problem 4.
\end{example}

\begin{figure}[t]
\centering
\includegraphics[scale=0.45]{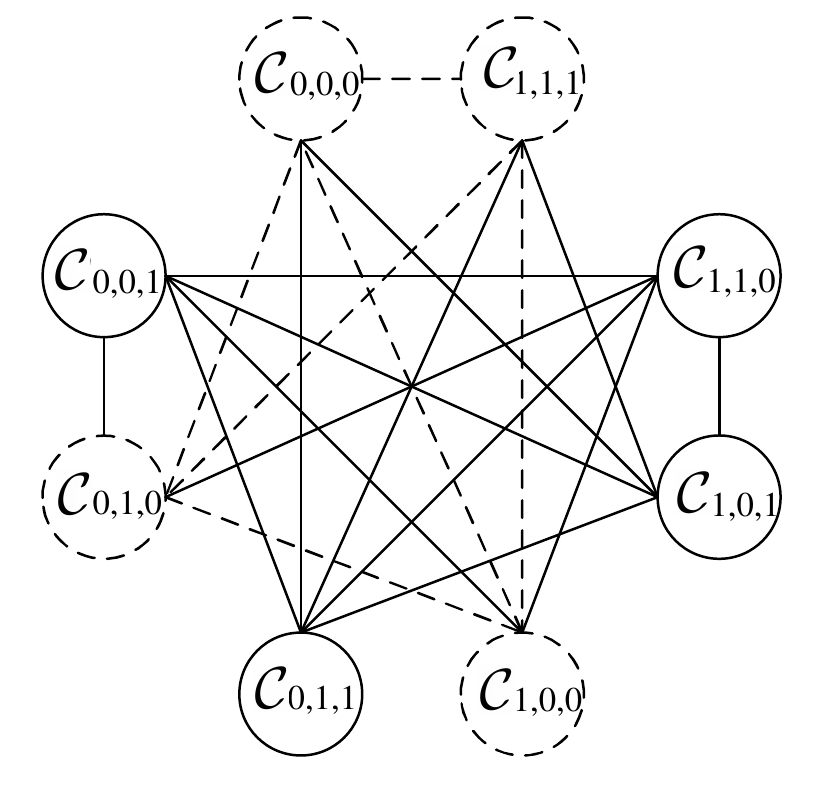}
\caption{The 4-partite graph $\mc{G}$ corresponds to \eqref{eqn: C11-C41}. }
\label{fig: 2-partite}
\vspace{-0.3cm}
\end{figure}

However, up to now, there is no polynomial time method available  for finding $k$-cliques in $k$-partite graphs\cite{T2002Finding, 2013Mirghorbani},  so does for  Problem 4. Then, there is also no polynomial time method for solving Problem 3.

\section{PCM Partition Algorithms}

For practical application, in this section we  propose two algorithms to solve Problem 3 by using enumerative and greedy search, respectively.

\emph{Enumerative algorithm:} As shown in Algorithm 2,
 firstly we fix $S$ and enumerate all the feasible solutions with $S$ according to Theorem \ref{theorem: feasible solution};
then, we go through all the possible $S$ to obtain all the feasible solutions to Problem 3;
finally, among them, the one with the minimal $\omega(\mb H_{\mc T_0})$ is the optimal solution.

For a given $S$, we have $|\mc{F}_S| = L^{MS}$, i.e., the $\omega(\cdot)$ function is computed for $L^{MS}$ times.
Computing $\omega(\cdot)$ has complexity $O(MNZ^2/L)$ since $\mb H_{\mc{T}_0}$  has $NZ$ columns, each of which has $MZ/L$ elements.
Thus, Algorithm \ref{algo: enumerative} has complexity $O(MNZ^2 L^{MS-1})$ for a given $S$.

The above enumerative algorithm may not work for a large scale case.
Thus, we present a greedy method (similar to that in  \cite{Cui09}) as an alternative.

\emph{Greedy algorithm:} As shown in Algorithm \ref{algo: greedy},
for a given $S$, the idea is to handle $\mc{C}_{0,0},  \ldots, \mc{C}_{M-1,S-1}$ one-by-one, corresponding to Line 4.
Then, for each $\mc{C}_{m,s}$, select a locally best $LS$-class $\mc{C}_{m,s, {\ell^*}}$ to add into $\mc{T}_0$, corresponding to Line 11.
The $\omega(\cdot)$ function is computed for $O(MSL)$ times.
As a result, Algorithm \ref{algo: greedy} has complexity $O(M^2 S N Z^2)$ for a given $S$.

\begin{table}[t!]
\begin{algorithm}[H]
	\normalsize
	\caption{Enumerative method for solving Problem 3}\label{algo: enumerative}		
    \begin{algorithmic}[1]
		\REQUIRE $\mb H$, $L$.
		\ENSURE $({S}^*, {\mc T^*_0})$.
%
       \STATE{\textbf{Initialization:}} {$\mc F=\emptyset$}.
       \FOR {each factor $S$ of $Z/L$}
\STATE {$\mc F_S= \{(S,\cup_{m\in [M),s \in [S)} \mc C_{m,s,\l_s}): \l_s ~\mathrm{ranges ~over~}[L)\}.$}
\STATE {$\mc F=\mc F \cup \mc F_S.$}
\ENDFOR
	\STATE {Return $({S}^*, {\mc T^*_0}) = \mathop{\arg\min}_{({S}, {\mc T_0}) \in \mc F} \omega(\mb H_{\mc T_0})$.}	
	\end{algorithmic}
\end{algorithm}
\vspace{-0.6cm}
\end{table}
	\begin{table}[t!]
\begin{algorithm}[H]
	\normalsize
	\caption{Greedy method for solving Problem 3}\label{algo: greedy}
	\begin{algorithmic}[1]
		\REQUIRE $\mb H$, $L$.
		\ENSURE $({S}^*, {\mc T^*_0})$.
    	\STATE{\textbf{Initialization:}} ${\omega^*} = \infty$.
    	\FOR {each factor $S$ of $Z/L$}
            \STATE{\textbf{Initialization:}} $\mc T_0 = \emptyset$.
        	\FOR {each $m \in [M)$ and $s \in [S)$}
                \STATE{\textbf{Initialization:}} $\omega^*_{m,s} = \infty$, ${l^*} = 0$.
                \FOR {each $l \in [L)$}
                    \IF {$ \omega(\mb H_{\mc T_0 \cup \mc{C}_{m,s, l}})< \omega^*_{m,s}$}

                    \STATE Let $(\omega^*_{m,s}, {l^*}) = (\omega(\mb H_{\mc T_0 \cup \mc{C}_{m,s, l}}), l)$.
                    \ENDIF
            	\ENDFOR
                \STATE  $\mc T_0 = \mc T_0 \cup \mc{C}_{m,s, {l^*}}$.
        	\ENDFOR
            \STATE  If $\omega(\mc T_0) < {\omega^*}$, let $({\omega^*}, {S^*},{\mc T^*_0}) = (\omega (\mb H_{\mc T_0}), S, \mc T_0)$.
    	\ENDFOR

	\end{algorithmic}
\end{algorithm}\vspace{-0.8cm}
\end{table}

The $\omega(\cdot)$ function is computed frequently in  Algorithms \ref{algo: enumerative} and \ref{algo: greedy}.
Reducing the complexity of computing $\omega(\cdot)$ can benefit the efficiency of both algorithms.
Specifically, since $\omega(\mb H_{\mc T_0})$ is of interest, it suffices to compute the weight of certain columns of $\mb H_{\mc T_0}$ instead of all its columns, as shown in the following theorem.

\begin{theorem}\label{theorem: column weight}
Assume that $(S,\mc T_0)$ is a solution to Problem 3.
Then,  \vspace{-0.1cm}
\begin{equation*}
  \omega(\mb H_{\mc T_0})= \max_{ j \in [NZ): j\bmod Z < LS} \omega_j(\mb H_{\mc T_0}).
  \vspace{-0.1cm}
\end{equation*}
\end{theorem}
\begin{proof}
Take an arbitrary row vector $\mb x$ from $\mb H_{\mc T_0}$, whose $j$-th entry is 1.
According to  Theorem \ref{theorem: feasible solution} that $\mc T_0$ consists of the $LS$-classes,
we then have $\phi^{LS}(\mb x)$ is also a row vector of $\mb H_{\mc T_0}$ with $\pi^{LS}(j)$-th entry being 1.
Analogously, there exist row vectors in $\mb H_{\mc T_0}$, whose $\pi^{2LS}(j)$-th, $\pi^{3LS}(j)$-th, $\ldots$, $\pi^{Z-LS}(j)$-th entries are 1, respectively.
That is, a 1-entry in $j$-th column of $\mb H_{\mc T_0}$ results in a 1-entry in $j'$-th column of $\mb H_{\mc T_0}$ respectively, for each $j' \in \{j\}_{LS}$.
Therefore,  $\omega_j(\mb H_{\mc T_0})=\omega_{j'}(\mb H_{\mc T_0})$ for each $j \in [NZ)$ and $j' \in\{j\}_{LS}$.
Noting that for any $j \in [NZ)$, there exists a $j' \in \{j\}_{LS}$ such that $j' \bmod Z < LS$, we have
\begin{eqnarray*}
\omega(\mb H_{\mc T_0})
&=& \max_{j \in [NZ)} \omega_j(\mb H_{\mc T_0})\\
&=& \max_{j \in [NZ): j \bmod Z < LS} \omega_j(\mb H_{\mc T_0}).
\end{eqnarray*}
The proof is completed.
\end{proof}

By Theorem \ref{theorem: column weight},  the complexity of computing $\omega(\cdot)$ can be reduced by a factor of $\frac {Z}{LS}$.
The reduction is significant since generally $L$ and $S$ are small while $Z$ is large for practical scenarios.

The following  lower bound on $\omega (\mb H_{\mc T_0})$ is obvious,
which would be greatly helpful for evaluating the quality of a feasible solution.

\begin{theorem}\label{theorem: lower bound}
For any  solution $(S,\mc T_0)$ to Problem 3,
\[
     \omega_{LB} {\triangleq } \lceil \omega(\mb{H})/L \rceil \leq \omega(\mb H_{\mc T_0}).
\]
\end{theorem}

Similar to Theorem \ref{theorem: lower bound}, we have the following upper bound on  layer distance.
\begin{theorem}\label{theorem: column weight with layer distance}
Given $\mb H$ and $L$, let $(S, \mc T_0)$ be an arbitrary feasible solution to Problem 3.
Then,
\[
     d_{UB} \triangleq  \lfloor L/\omega(\mb{H}) \rfloor \geq d(S,\mc T_0).
\]
\end{theorem}

\begin{proof}
Recall that $\mb H^{(S,d(S,\mc T_0))}=\mb H+\phi^S(\mb H)+\phi^{2S}(\mb H)+\cdots+\phi^{(d(S,\mc T_0)-1)S}(\mb H)$.
 According to Theorem \ref{theorem: lower bound},  we have
\begin{align}
1=\omega\left(\mb H^{\left(S, d\left(S,\mc T_0\right)\right)}_{\mc T_0}\right)  &\geq  \omega \left(\mb H^{\left(S, d\left(S,\mc T_0\right)\right)}\right)/L \notag \\
&= d(S,\mc T_0) \omega(\mb H) / L, \label{eqn:d_UB}
\end{align}
since $\omega (\mb H)=\omega (\phi^{iS}(\mb H))$, $1\le i< d(S,\mc T_0)$, which leads to the claimed bound.
\end{proof}

Further, when considering a desired layer distance and finding a proper $L$, we have the following corollary.

\begin{corollary}\label{cor:L_LB}
Given $\mb H$, to have a partition scheme with layer distance at least $k$, then   $L\ge L_{LB}$ where the lower bound $L_{LB}$ is the smallest factor of $Z$ and satisfies $L_{LB} \geq k \cdot \omega(\mb H)$.
\end{corollary}
\begin{proof}
By replacing $d(S,\mc T_0)$ in \eqref{eqn:d_UB} with $k$, we have
$$L \geq k \cdot \omega(\mb H). $$
Then, the corollary follows from the fact  $L\mid Z$.
\end{proof}

\section{{PCM Design for Partitioning}}

According to the algorithms presented in Section V, we can get a partition scheme with a small maximum column weight or a large layer distance.
However, for some cases, finding such solution may be too time-consuming or even there is no solution achieving $\omega_{LB}$ or the desired layer distance.


In this section, we design such QC-LDPC codes based on QC-PEG algorithms \cite{Li04QCLDPC,He18PEG}.
In particular, we first introduce the framework of QC-PEG algorithms.
Then, we derive a condition for PCM to guarantee that the PCM has a partition scheme attaining $\omega_{LB}$ or the desired layer distance.
Finally, we add this condition into the CN selection strategy of a QC-PEG algorithm to construct the desired QC-LDPC codes.

\subsection{Framework of QC-PEG}
QC-PEG algorithms are widely used to design QC-LDPC codes by avoiding short cycles.
To illustrate the QC-PEG algorithms, we give two notations in the following.
\begin{itemize}
\item Denote $\left(c_{i}, v_{j}\right)_{Z}=\{\left(c_{\pi^t(i)}, v_{\pi^t(j)}\right):
     t \in [Z)\},$ for each $i \in [MZ), j \in [NZ)$ as the cyclic edge set (CES) which $\left(c_{i}, v_{j}\right)$ belongs to.
     For a QC-LDPC PCM, we have $\left(c_{i}, v_{j}\right) \in \mc E \Longleftrightarrow\left(c_{i}, v_{j}\right)_{Z} \subseteq \mc E$.
\item Denote $ \mathbf{d}=(d_0,d_1,\ldots,d_{NZ-1})$  the  VN degree sequence, where $d_j$ is the degree of the $j$-th VN. Clearly,  $d_j = d_{j'}$ for $j, j' \in [NZ)$ with $\lfloor j/Z \rfloor = \lfloor j'/Z \rfloor.$
\end{itemize}

The framework of QC-PEG algorithms is shown in Algorithm 4, where the key point is to design an efficient strategy for selecting CNs, corresponding to
Line 4.
A CN selecting strategy consists of several criteria, which are carried out in order.
For instance, a QC-PEG algorithm was illustrated in \cite{He18PEG}, whose CN selecting strategy is outlined in Strategy 1.

\textbf{\emph{Strategy 1:}}

1) Select the CNs $c_i$ with $(c_i, v_j) \notin \mc E$.

2) Select the survivors $c_i$ such that the length of the shortest cycle passing through $c_i$ and $v_j$ in $(\mc V, \mc E \cup (c_i,v_j)_Z)$ is maximum. (If no such cycles exist, we define the length as infinite.)

3) Select the survivors with the minimal degree.

4) Select a survivor randomly.

\begin{table}[t!]
\begin{algorithm}[H]
	\normalsize
	\caption{Framework of QC-PEG algorithms}\label{algo: PEG}		
    \begin{algorithmic}[1]
		\REQUIRE $M, N, L, Z, \mathbf{d} .$
		\ENSURE $\mc G .$
		\STATE{\textbf{Initialization:}} {$\mc G=(\mc V, \emptyset) .$}
          \FOR {$j = 0, Z, ..., (N-1)Z$}
          \FOR {each $t = 0, 1, ..., d_{j }- 1$}
                \STATE  Select a CN $c_i$ based on a certain strategy.
                \STATE $\mc E =\mc E \cup \left(c_{i}, v_{j}\right)_{Z}$.
            \ENDFOR
            \ENDFOR
            \RETURN $\mc G $.
	\end{algorithmic}
\end{algorithm}
\vspace{-0.6cm}
\end{table}

\subsection{PCM Design for Achieving $\omega_{LB}$}
In this subsection, for given $M, N, Z, L$ and $\mb{d}$, we employ Algorithm \ref{algo: PEG} to construct $\mb H$ achieving $\omega_{LB} = \lceil \omega(\mb H)/L \rceil = \max_{j \in [NZ)}  \lceil d_j / L \rceil$ with a straightforward partition scheme $(S, \mc T_0) = (1,\cup_{m \in [M)} \mc C_{m,0,0})$.
We consider $S = 1$ since it is the simplest and further it is  always sufficient for the construction.

\begin{lemma}\label{lemma: design}
 $(S, \mc T_0) = (1,\cup_{m \in [M)} \mc C_{m,0,0})$ is an optimal solution with $\omega(\mb H_{\mc T_0}) = \omega_{LB}$ iff
\begin{equation}\label{eqn: design}
|\{x \in \mc N(v_j): x \equiv 0 \bmod L \}| \leq \omega_{LB}, \forall j \in [NZ)
\end{equation}
where the equality holds for some $j$s.
\end{lemma}
\begin{proof}
Noting that $\cup_{m \in [M)} \mc C_{m,0,0}=\{0,L,2L,\ldots,MZ-L\}=\{x \in [MZ): x \equiv 0 \bmod L\}$,
for each $j \in [NZ)$, we have $\omega_j(\mb H_{\mc T_0})= |\mc N(v_j) \cap \mc T_0|=|\mc N(v_j) \cap (\cup_{m \in [M)} \mc C_{m,0,0})|=|\{x \in \mc N(v_j): x \equiv 0 \bmod L \}|$.
Then, the lemma follows from the fact that $\omega(\mb H_{\mc T_0})= \max_{j \in [NZ)}\omega_j(\mb H_{\mc T_0})$.
\end{proof}

Based on  Lemma  \ref{lemma: design}, we modify the selection strategy of the QC-PEG algorithm \cite{He18PEG} to obtain a PCM $\mb H$ with a straightforward partition scheme as follows.

\textbf{\emph{Strategy 2:}}

1) Select the CNs $c_i$ such that  $(c_i, v_j) \notin \mc E$ and
\begin{equation}\label{eqn:strategy2}
|\{x \in \mc N'(v_{j+z}): x \equiv 0 \bmod L\}| \leq \omega_{LB}, \forall z \in [Z),
\end{equation}
where $\mc N'(v_{j+z}) \triangleq \mc N(v_{j+z}) \cup \{\pi^z(i)\}$ denotes the index set of CNs connected to $v_{j+z}$ after the CES $(c_i,v_j)_Z$ is added in $\mc E$.

2) The same as criteria 2 to 4 in Strategy 1.

According to criterion 1, with $j$ going through $0,Z,\dots,(N-1)Z$, the PCM constructed by Strategy 2 satisfies \eqref{eqn: design} and has a partition scheme achieving $\omega_{LB}$ consequently.

Compared with Strategy 1, the criterion 1 in Strategy 2 is more restrictive on CNs.
For example, assume $\omega_{LB}=1$. When we select a $t$-th CN  connecting to $v_j$, there are $MZ-t$ CNs surviving from criterion 1 in Strategy 1, since there are $t$ CNs already connected to $v_j$ and  then are unavailable.
On the other hand, there are $MZ-t(MZ/L)$ candidate CNs after performing the criterion 1 in Strategy 2, since there are $t$ CNs already connected to $v_j$, each of which results in $MZ/L$ unavailable CNs.

Fortunately, the following theorem shows that criterion 1 of Strategy 2 does not miss any $\mb H$ achieving $\omega_{LB}$ under possible row permutation.

\begin{theorem}\label{theorem: design principle}
$\mb H$ has a partition with $\omega(\mb H_{\mc{T}_0})= \omega_{LB}$ iff there exists a QC-LDPC PCM
\begin{eqnarray*}\label{eqn: H prime1}
	\mathbf{H'}=\left[\begin{array}{c}
	 \phi ^{j_0}(\mb H_0) \\
     \phi ^{j_1}(\mb H_1) \\
	\vdots\\
	 \phi ^{j_{M-1}}(\mb H_{M-1}) \\
	\end{array}\right], j_0, j_1,..., j_{M-1} \in [L),
	\end{eqnarray*}
satisfying \eqref{eqn: design}.
\end{theorem}

\begin{proof} Let $\mc T_0= \cup_{m \in [M)} \mc C_{m,0,j_m}$ and
 $\mc T'_0= \cup_{m \in [M)} \mc C_{m,0,0}$. The  proof is based the fact
 $\mb H_{\mc T_0}=\mb H'_{\mc T'_0}$.

\emph{Sufficiency:} Assume that $\mb H'$  satisfies \eqref{eqn: design}.
According to Lemma \ref{lemma: design},  $(1, \mc T'_0) $ is an optimal solution with $\omega(\mb H'_{\mc T'_0}) = \omega_{LB}$.
Then, we have $\omega(\mb H_{\mc T_0})=\omega(\mb H'_{\mc T'_0})=\omega_{LB}$.

\emph{Necessity:}
Suppose that $\mb H$ has a partition $(1, \mc T_0)$ with $\omega(\mb H_{\mc{T}_0})= \omega_{LB}$.
That is, $(1, \mc T'_0) $ is an optimal solution with $\omega(\mb H'_{\mc T'_0}) =\omega(\mb H_{\mc{T}_0})= \omega_{LB}$.
According to Lemma \ref{lemma: design}, \eqref{eqn: design} is satisfied for $\mb H'$.
\end{proof}

Moreover, we can design a PCM achieving the lower bound for arbitrary VN degree sequence $\mb d$.
\begin{theorem}\label{Theorem: column }
Given an arbitrary $\mb d$, there always exists an $\mb H$ constructed by Strategy 2, achieving $\omega_{LB}$.
\end{theorem}

\begin{proof}
For each $j \in \{0,Z,\dots,(N-1)Z\}$ and $z\in [Z)$, we have $x \in \mc N(v_{j})$ iff $\pi^z(x) \in \mc N(v_{j+z})$ since $\mb H$ is a QC-LDPC PCM.
Therefore, we have $|\{x \in \mc N(v_{j+z}) : x \equiv 0 \bmod L\}|$=$|\{x \in \mc N(v_{j}) : x \equiv -z \bmod L\}|$.
Accordingly, to make sure that \eqref{eqn:strategy2} is satisfied, we just need
\begin{equation}\label{eqn:design v_j}
|\{x \in \mc N(v_{j}) : x \equiv -z \bmod L\}| \leq w_{LB}, \forall z \in [Z).
\end{equation}
On the other hand,  it is easy to verify that  $\mc N(v_{j})=\cup_{l \in [L)}\{l+0,l+Z,l+2Z,\ldots,l+(\omega_{LB}-1)Z\}$
satisfies \eqref{eqn:design v_j} with  $L\cdot \omega_{LB} $ elements.
According to Theorem \ref{theorem: lower bound}, $L\cdot \omega_{LB} \geq \omega(\mb H) \geq d_j$, for each $j \in [NZ)$.
Therefore, there always exist candidate CNs after performing criterion 1 in Strategy 2.
\end{proof}

\subsection{PCM Design for Achieving Desired Layer Distance}
In this subsection, we construct a PCM with desired layer distance $k$.
We focus on $k>0$, i.e. $\omega(\mb H_{\mc T_0})=1$.

\begin{lemma}\label{lemma: design with layer distance}
$(S, \mc T_0) = (1,\cup_{m \in [M)} \mc C_{m,0,0})$ is a feasible solution with layer distance at least $k>0$ iff
for any $j \in [NZ)$,
\begin{equation}\label{eqn: design k}
\left\{\begin{array}{l}
\mc N(v_j) \cap \mc N(v_{\pi^t(j)}) = \emptyset, \forall t \in [1, k) ,\\
| \{x \in \cup_{t \in [k)} \mc N(v_{\pi^t(j)}): x \equiv 0 \bmod L\} |  \leq 1.
\end{array}\right.
\end{equation}
%
%
\end{lemma}

\begin{proof}
\emph{Sufficiency: }Assume that \eqref{eqn: design k} is satisfied.
Recall that $\mb H^{(1,k)}=\mb H+\phi^1(\mb H)+\phi^{2 }(\mb H)+\cdots+\phi^{k-1}(\mb H)$.
Since $\mc N(v_j) \cap \mc N(v_{\pi^t(j)}) = \emptyset, \forall t \in [1, k)$, $\mb H^{(1,k)}$ is binary.
For each $j\in [NZ)$, we then have $\omega_j(\mb H^{(1,k)}_{\mc T_0})=| {x \in \cup_{t \in [k)} \mc N(v_{\pi^t(j)}): x \equiv 0 \bmod L} |  \leq 1$, implying $\omega(\mb H^{(1,k)}_{\mc T_0})=1$. Thus, $d(1,\mc T_0)\geq k$ by Theorem \ref{theorem: conflict}.

\emph{Necessity:} Assume that $(S, \mc T_0) = (1,\cup_{m \in [M)} \mc C_{m,0,0})$ is a feasible solution with layer distance at least $k$.
According to Theorem \ref{theorem: conflict}, we have $\mb H^{(1,k)}$ is binary and $\omega(\mb H^{(1,k)}_{\mc T_0})=1$.
Therefore, we have $\mc N(v_j) \cap \mc N(v_{\pi^t(j)}) = \emptyset, \forall t \in [1, k)$ and for each $j\in [NZ)$, $| {x \in \cup_{t \in [k)} \mc N(v_{\pi^t(j)}): x \equiv 0 \bmod L} | =\omega_j(\mb H^{(1,k)}_{\mc T_0})\leq 1$.
\end{proof}

By means of Lemma \ref{lemma: design with layer distance}, we can obtain a PCM with desired layer distance $k$ by the following Strategy:

\textbf{\emph{Strategy 3:}}

1) Select the CNs $c_i$ such that  $(c_i, v_j) \notin \mc E$ and
\begin{equation}
\left\{\begin{array}{l}
\mc N'(v_{j+z}) \cap \mc N'(v_{\pi^t(j+z)}) = \emptyset, \forall t \in [1, k) ,\\
| \{x \in \cup_{t \in [k)} \mc N'(v_{\pi^t(j+z)}): x \equiv 0 \bmod L\} |  \leq 1,
\end{array}\right.
\end{equation}
for each $z \in [Z)$, where  $\mc N'(v_{j+z})\triangleq \mc N(v_{j+z}) \cup \{\pi^z(i)\}$ and $\mc N'(v_{\pi^t(j+z)})\triangleq\mc N(v_{\pi^t(j+z)}) \cup \{\pi^{z+t}(i)\}$ denote the index sets of CNs connected to $v_{j+z}$ and $v_{\pi^t(j+z)}$ after the CES $(c_i,v_j)_Z$ is added in $\mc E$, respectively.

2) The same as criteria 2 to 4 in Strategy 1.

Analogous to Theorems \ref{theorem: design principle} and \ref{Theorem: column }, we have the following results.

\begin{theorem}
$\mb H$ has a partition with layer distance at least $k$ iff there exists a QC-LDPC PCM
\begin{eqnarray}
	\mathbf{H'}=\left[\begin{array}{c}
	 \phi ^{j_0}(\mb H_0) \\
     \phi ^{j_1}(\mb H_1) \\
	\vdots\\
	 \phi ^{j_{M-1}}(\mb H_{M-1}) \\
	\end{array}\right], j_0, j_1,..., j_{M-1} \in [L),
	\end{eqnarray}
satisfying \eqref{eqn: design k}.
\end{theorem}

\begin{theorem}
Given an arbitrary VN degree sequence $\mb d$ and any desired layer distance $k \leq d_{UB}$, there exists an $\mb H$ constructed by Strategy 3, achieving layer distance $k$.
\end{theorem}

In practice, for given $L$ and $\mb d$, noting that $\omega(\mb H)=\max_{j \in [NZ)}d_j$, the upper bound of layer distance $d_{UB}$ can be computed according to Theorem \ref{theorem: column weight with layer distance}.
Then, based on the value $d_{UB}$, we choose different strategies to construct a PCM for given $M,N$ and $Z$.
More specifically, if $d_{UB}< 1$, we have $\omega_{LB} > 1$ and there does not exist PCM avoiding data dependency within a layer.
In this case, we target to  construct a PCM by Strategy 2 to achieve $\omega_{LB}$.
For the case $d_{UB} \geq 1$, we can construct a PCM  by Strategy 3, which aims to achieve $\omega_{LB}=1$ and layer distance $1\leq k \leq d_{UB}$.
We remark that a larger desired layer distance $k$ leads to shorter computation delay but fewer options in code construction which may affect error-correction performance.
Thus, we should set $k$ according to the requirement of practice.

\section{Simulation Results}
\begin{table}[h]
	\small
	\renewcommand{\arraystretch}{1.2}
	\caption{Parameters for the 5G LDPC Codes}
	\label{table: codes}
	\centering
	\setstretch{1}
	\begin{tabular}{ccccc}
		\toprule
		PCMs & $M$  & $N$  & $Z$      & $\omega(\mb{H})$ \\ \midrule
		1                                                       & 5  & 27 & 384 & 5      \\
		2                                                       & 46 & 68 &    384                  & 30     \\
		3                                                       & 7  & 17 & 112 & 6      \\
		4                                                       & 17 & 27 &    112                  & 13     \\
		5                                                       & 42 & 52 &    112                  & 23     \\
		\bottomrule
	\end{tabular}
	\vspace{-0.2cm}
\end{table}

In this section, we evaluate the validity of our proposed PCM partition methods for the 5G LDPC codes \cite{5gChannel1}.
The key parameters of the PCMs of these codes are listed in Table \ref{table: codes}.

\begin{table}[H]
\renewcommand{\arraystretch}{1.2}
	\caption{PCM Partition Results for the 5G LDPC Codes}
	\label{table: simulation1}
	\centering
	\small
	\setstretch{0.95}
	\begin{tabular}{|c|c|c|c|c|}
		\hline
		\multirow{2}{*}{PCMs} & \multirow{2}{*}{$L$} & \multirow{2}{*}{$\omega _{LB}$} & \multicolumn{2}{c|}{$\omega(\mb H_{\mc T^*_0}), S^*$} \\ \cline{4-5}
		&                    &                            & Algorithm \ref{algo: enumerative} & Algorithm  \ref{algo: greedy}  \\ \hline
		\multirow{15}{*}{1} & 2   & 3     & 3, 1        & 3, 1             \\
		& 3   & 2        	& $-$   & $\textbf{3}, 1$     \\
		& 4   & 2    & 2, 1        			& 2, 1             \\
		& 6  & 1      	& $-$    & $\textbf{2}, 1$    \\
		& 8  & 1       	& $-$     & $\textbf{2}, 1$  \\
		& 12 & 1      	& 1, 4       & $\textbf{2}, 1$       \\
		& 16 & 1     & 1, 1        			& 1, 1             \\
& 24 & 1     & 1, 1        			& 1, 1             \\
& 32 & 1     & 1, 1        			& 1, 1             \\
& 48 & 1     & 1, 1        			& 1, 1             \\
& 64 & 1     & 1, 1        			& 1, 1             \\
& 96 & 1     & 1, 1        			& 1, 1             \\
& 128 & 1     & 1, 1        			& 1, 1             \\
& 192 & 1     & 1, 1        			& 1, 1             \\
& 384 & 1     & 1, 1        			& 1, 1             \\ \hline
		\multirow{15}{*}{2} & 2  & 15    & 15, 1       & 15, 1            \\
		& 3  & 10    & 10, 1       & 10, 1            \\
		& 4  & 8     & 8, 1        & 8, 1             \\
		& 6  & 5     & 5, 1        & 5, 1             \\
		& 8  & 4     & 4, 1        & 4, 1             \\
		& 12 & 3     & 3, 1        & 3, 1             \\
		& 16 & 2     & 2, 1        & 2, 1             \\
		& 24 & 2     & 2, 1        & 2, 1             \\
		& 32 & 1     & 1, 1        & 1, 1             \\
	& 48 & 1     & 1, 1        			& 1, 1             \\
& 64 & 1     & 1, 1        			& 1, 1             \\
& 96 & 1     & 1, 1        			& 1, 1             \\
& 128 & 1     & 1, 1        			& 1, 1             \\
& 192 & 1     & 1, 1        			& 1, 1             \\
& 384 & 1     & 1, 1        			& 1, 1             \\ \hline
		\multirow{9}{*}{3}   & 2                      & 3          & 3, 1    & 3, 1             \\
		& 4                  & 2                      & 2, 1          & 2, 1             \\
		& 7                  & 1                               & 1, 4    & \textbf{2}, 1          \\
		& 8                  & 1                               & 1, 1       & \textbf{2}, 1       \\
		& 14                 & 1                      & 1, 1          & 1, 1             \\
& 16                 & 1                      & 1, 1          & 1, 1             \\
& 28                 & 1                      & 1, 1          & 1, 1             \\
& 56                 & 1                      & 1, 1          & 1, 1             \\
& 112                 & 1                      & 1, 1          & 1, 1             \\ \hline
		\multirow{9}{*}{4}   & 2                      & 7          & 7, 1          & 7, 1             \\
		& 4                  & 4                      & 4, 1          & 4, 1             \\
		& 7                  & 2                      & 2, 1          & 2, 1             \\
		& 8                  & 2                      & 2, 1          & 2, 1             \\
		& 14                 & 1                      & 1, 1           & \textbf{2}, 1   \\
		& 16                 & 1                       & 1, 1         & \textbf{2}, 1    \\
		& 28                 & 1                      & 1, 1          & 1, 1             \\
& 56                 & 1                      & 1, 1          & 1, 1             \\
& 112                 & 1                      & 1, 1          & 1, 1             \\ \hline

		\multirow{9}{*}{5}   & 2                      & 12         & 12, 1         & 12, 1            \\
		& 4                  & 6                      & 6, 1          & 6, 1             \\
		& 7                  & 4                      & 4, 1          & 4, 1             \\
		& 8                  & 3                      & 3, 1          & 3, 1             \\
		& 14                 & 2                      & 2, 1          & 2, 1             \\
		& 16                 & 2                      & 2, 1          & 2, 1             \\
		& 28                 & 1                      & 1, 1          & 1, 1             \\
& 56                 & 1                      & 1, 1          & 1, 1             \\
& 112                 & 1                      & 1, 1          & 1, 1             \\ \hline
	\end{tabular}
\vspace{-0.3cm}
\end{table}



\begin{table}[H]
\renewcommand{\arraystretch}{1.2}
	\caption{PCM Partition Results with Desired Layer Distance for the 5G LDPC Codes}
	\label{table: simulation2}
	\centering
	\small
	\setstretch{0.95}
	\begin{tabular}{|c|c|c|c|c|}
		\hline
		\multirow{2}{*}{PCMs} & \multirow{2}{*}{$k$} & \multirow{2}{*}{$ L_{LB}$} & \multicolumn{2}{c|}{$L^*, S^*$} \\ \cline{4-5}
		&                    &                        & Algorithm   \ref{algo: enumerative}   & Algorithm  \ref{algo: greedy}  \\ \hline
		\multirow{3}{*}{1} & 2   & 12     & 24, 1        & 36, 1             \\
		& 3   & 16    & $32, 1$     	& $48,1$       \\
		& 4   & 24    & 64, 1        			& 64, 3             \\\hline
		\multirow{3}{*}{2} & 2  & 64    & 96, 1       & 96, 1            \\
		& 3  & 96    & 96, 1       & 128, 1            \\
		& 4  & 128     & 192, 1        & 192, 1             \\\hline
	
		\multirow{3}{*}{3}   & 2                      & 14          & 28, 1    & 28, 1             \\
		& 3                  & 28                      & 28, 2          & 56, 1             \\
		& 4                  & 28                      & 56, 1          & 56, 1             \\\hline
		\multirow{3}{*}{4}   & 2                      & 28          & 28, 1          & 56, 1             \\
		& 3                  & 56                      & 56, 1          & 56, 1             \\
		& 4                  & 56                      & 112, 1          & 112, 1             \\
\hline
		\multirow{3}{*}{5}   & 2                      & 56   & 56, 1         & 112, 1            \\
		& 3                  & 112                      & 112, 1          & 112, 1             \\
		& 4                  & 112                      & $-$          & $-$             \\\hline
	\end{tabular}
\vspace{-0.3cm}
\end{table}

We apply both (enumerative) Algorithm  \ref{algo: enumerative} and (greedy) Algorithm \ref{algo: greedy}
to obtain the optimal or locally optimal solution $({S^*}, \mc T^*_0)$ to Problem 3,
which may be early terminated once a solution $(S, \mc T_0)$ with $\omega(\mb H_{\mc T_0}) = \omega _{LB}$ is found or the running time exceeds a preset threshold.
Here, we go through all the possible $S$, i.e. all the factors of $Z$.
The partition results are presented in Table \ref{table: simulation1}.

In Table \ref{table: simulation1}, we use the bold face to highlight the $\omega(\mb H_{\mc T^*_0})$ which is greater than $\omega_{LB}$.
We use `$-$' to indicate the case where Algorithm \ref{algo: enumerative} cannot
achieve the lower bound $\omega_{LB}$ within a preset running time limit.

We can see that both algorithms can generally achieve $\omega_{LB}$ or differ by one.
There exist cases (e.g., $L=12$ for PCM 1) where $\omega_{LB}$ is achievable only for $S^* > 1$.
This situation has never been considered in the literature.
Moreover, there exist cases (e.g., $L=8$ for PCM 3) where Algorithm \ref{algo: greedy}  cannot achieve $\omega_{LB}$ while Algorithm \ref{algo: enumerative} can, as they work in a greedy and enumerative fashion, respectively.
However, there also exist cases (e.g., $L=6$ for PCM 1) where Algorithm \ref{algo: enumerative} may exceed the preset time limit due to large search space.

Table \ref{table: simulation2} shows the partition results with desired layer distance for the 5G LDPC codes.
We use $k$ to represent the desired layer distance and $L^*$ to represent the minimum number of layers that may have a scheme with desired layer distance.
We use `$-$' to indicate the case where the algorithm cannot find a feasible solution with desired $k$ within a preset running time limit.
Accordingly, we also list $ L_{LB}$, which is the lower bound of $L^*$ given in Corollary \ref{cor:L_LB}.
We can see that in most cases, both algorithms can find a partition scheme with desired layer distance, where the minimum number of layers $L^*$ is slightly greater than or equal to $L_{LB}$. Since they work in an enumerative and greedy fashion, respectively,
there exist cases (e.g., $k=2$ for PCM 1) where Algorithm \ref{algo: enumerative} can find a desired solution under a lower $L^*$ than that in Algorithm \ref{algo: greedy}.
In addition, there exists a case ($k=4$ for PCM 5) where both algorithms cannot find a partition scheme with desired layer distance.

\begin{figure}[t]
\centering
\includegraphics[scale=0.5]{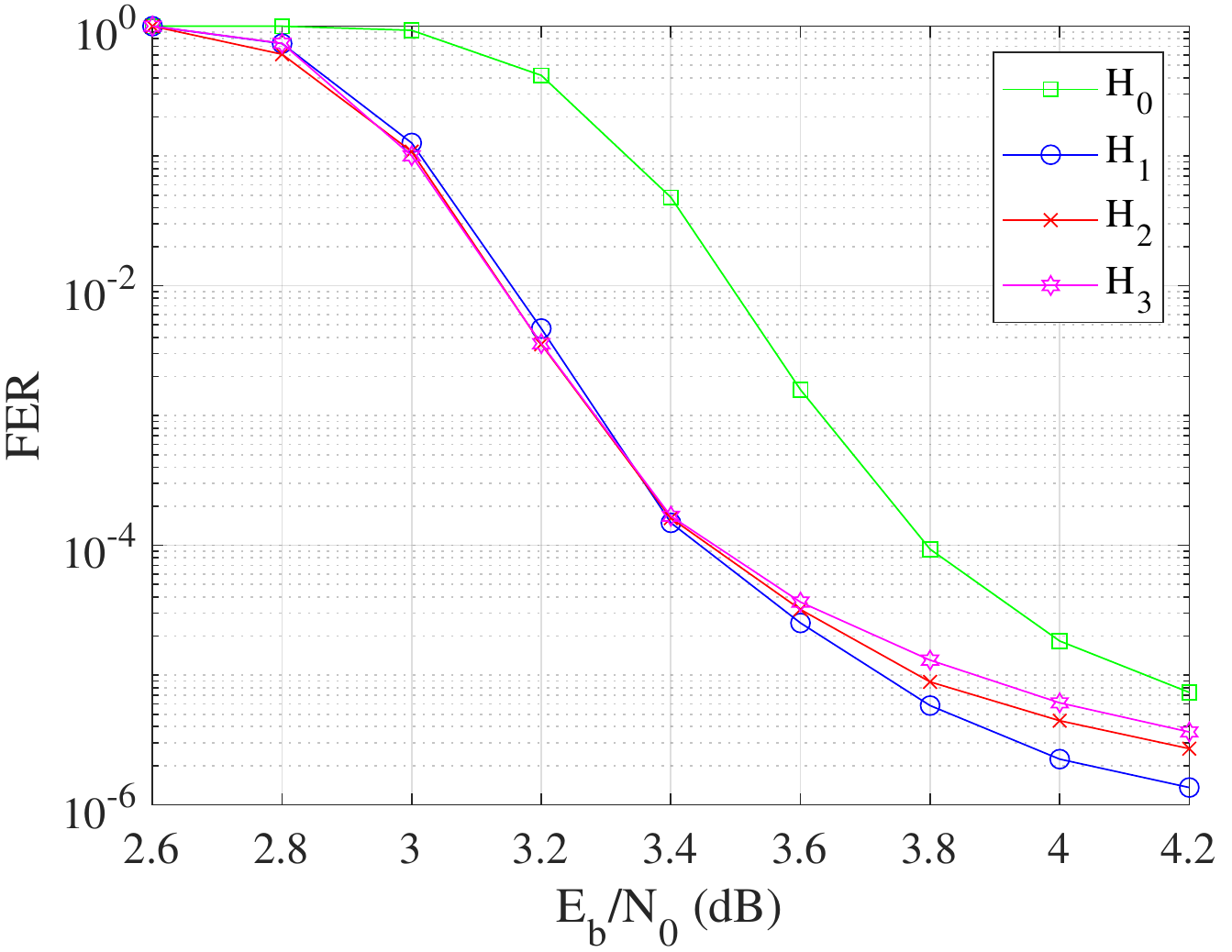}
\caption{Frame error rate (FER) performance of different QC-LDPC codes.}\label{fig: FER}
\vspace{-0.3cm}
\end{figure}

\begin{figure}[t]
\centering
\includegraphics[scale=0.5]{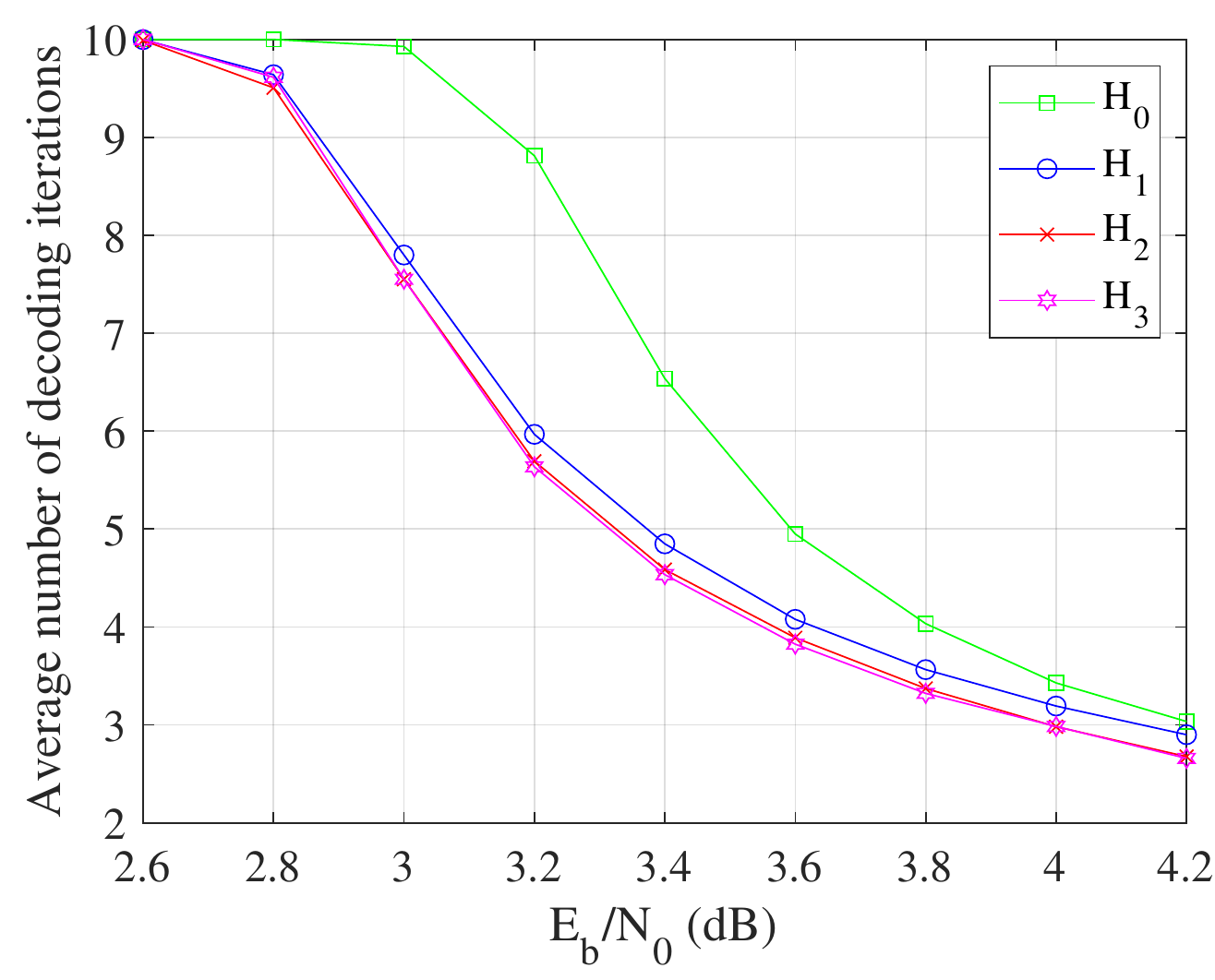}
\caption{Average number of decoding iterations of different QC-LDPC codes.}\label{fig: interation}
\vspace{-0.3cm}
\end{figure}

As can be seen from Tables III and IV, there exist cases that $\omega_{LB}$ or the desired layer distance is not achieved for PCM 1, denoted as $\mb H_{0}$.
We use the modified QC-PEG algorithm to construct two QC-LDPC codes to achieve $\omega_{LB}$ and the desired layer distance respectively.
More specifically, at $L = 6$, $\omega_{LB}$ is not archived, and then we construct a PCM to achieve $\omega_{LB}$ based on Strategy 2,  whose base matrix is denoted as $\mb H_{2}$;
at $L = L_{LB}=12$, the desired layer distance $k=2$ is not archived, and then we construct a PCM to achieve $k$ based on Strategy 3,  whose base matrix is denoted as $\mb H_{3}$.
As a comparison, we construct a PCM based on Strategy 1, whose base matrix is denoted as $\mb H_{1}$.
The base matrices of
$\mb H_{0}$, $\mb H_{1}$, $\mb H_{2}$, and $\mb H_{3}$
are presented in \eqref{eqn:B1}-\eqref{eqn:B4}, respectively.
In contrast to $\mb H_{0}$ used in the 5G standard, we remark that both $\mb H_{2}$ and $\mb H_{3}$ have a partition scheme with either lower maximum column weight or
larger layer distance and keep the same dimension, lifting size and VN degree sequence $\mb d$.

\begin{figure*}[!t]
\begin{small}
\begin{equation}\label{eqn:B1}
\setstretch{0.8}
\setlength{\arraycolsep}{1pt}
\mb B_0=\left[\begin{array}{lllllllllllllllllllllllllllll}
307 & 19 & 50 & 369 & -1 & 181 & 216 & -1 & -1 & 317 & 288 & 109 & 17 & 357 & -1 & 215 & 106 & -1 & 242 & 180 & 330 & 346 & 1 & 0 & -1 & -1 & -1 \\
76 & -1 & 76 & 73 & 288 & 144 & -1 & 331 & 331 & 178 & -1 & 295 & 342 & -1 & 217 & 99 & 354 & 114 & -1 & 331 & -1 & 112 & 0 & 0 & 0 & -1 & -1 \\
205 & 250 & 328 & -1 & 332 & 256 & 161 & 267 & 160 & 63 & 129 & -1 & -1 & 200 & 88 & 53 & -1 & 131 & 240 & 205 & 13 & -1 & -1 & -1 & 0 & 0 & -1 \\
276 & 87 & -1 & 0 & 275 & -1 & 199 & 153 & 56 & -1 & 132 & 305 & 231 & 341 & 212 & -1 & 304 & 300 & 271 & -1 & 39 & 357 & 1 & -1 & -1 & 0 & -1 \\
332 & 181 & -1 & -1 & -1 & -1 & -1 & -1 & -1 & -1 & -1 & -1 & -1 & -1 & -1 & -1 & -1 & -1 & -1 & -1 & -1 & -1 & -1 & -1 & -1 & -1 & 0
\end{array}\right].
\end{equation}
\end{small}
\end{figure*}

\begin{figure*}[!t]
\begin{small}
\begin{equation}\label{eqn:B2}
\setstretch{0.8}
\setlength{\arraycolsep}{1pt}
\mb B_1=\left[\begin{array}{llllllllllllllllllllllllllll}
-1 & 178 & -1 & -1 & 184 & -1 & 268 & -1 & 343 & 368 & -1 & 191 & 144 & -1 & 367 & 112 & 21 & -1 & 193 & -1 & -1 & 177 & 63 & 365 & -1 & 36 & 302 \\
47 & -1 & -1 & -1 & 197 & 241 & 349 & 13 & -1 & -1 & 35 & 173 & -1 & 189 & -1 & 367 & -1 & 186 & 89 & -1 & 360 & 156 & -1 & 316 & -1 & -1 & 14 \\
-1 & -1 & 354 & 128 & -1 & -1 & 311 & 116 & 286 & -1 & -1 & 238 & 69 & 27 & -1 & 198 & -1 & 324 & -1 & 361 & 262 & -1 & 105 & -1 & 165 & 173 & 59 \\
-1 & 176 & -1 & 277 & -1 & 26 & -1 & 252 & -1 & 243 & 279 & -1 & 202 & -1 & 277 & -1 & 189 & -1 & 378 & 322 & -1 & 306 & 158 & -1 & 267 & 67 & 56 \\
-1 & -1 & 189 & -1 & 125 & 10 & -1 & -1 & 5 & 265 & 269 & -1 & -1 & 20 & 56 & -1 & 37 & 130 & -1 & 122 & 31 & -1 & -1 & 76 & 88 & 59 & 123
\end{array}\right].
\end{equation}
\end{small}
\end{figure*}

\begin{figure*}[!t]
\begin{small}
\begin{equation}\label{eqn:B3}
\setstretch{0.8}
\setlength{\arraycolsep}{1pt}
\mb B_2=\left[\begin{array}{llllllllllllllllllllllllllll}
-1 & 253 & -1 & 213 & -1 & 263 & -1 & -1 & 79 & 303 & -1 & 97 & 50 & 23 & -1 & 173 & -1 & 373 & 53 & -1 & 38 & -1 & 93 & 154 & -1 & -1 & 221 \\
-1 & -1 & 141 & 112 & -1 & -1 & 33 & 379 & -1 & 12 & 354 & -1 & 157 & -1 & 312 & -1 & 218 & -1 & 344 & 108 & -1 & 320 & -1 & 126 & 334 & 357 & 370 \\
-1 & -1 & 376 & -1 & 364 & 242 & -1 & 10 & 65 & -1 & -1 & 195 & 315 & 242 & -1 & 219 & 107 & -1 & -1 & 272 & 52 & -1 & 331 & 253 & -1 & 161 & 348 \\
-1 & 40 & -1 & -1 & 306 & -1 & 144 & 222 & -1 & 107 & 11 & -1 & -1 & 104 & 106 & 55 & -1 & 52 & 355 & -1 & 257 & 202 & -1 & -1 & 200 & 338 & 19 \\
301 & -1 & -1 & -1 & 280 & 105 & 210 & -1 & 212 & -1 & 272 & 372 & -1 & -1 & 59 & -1 & 45 & 153 & -1 & 345 & -1 & 6 & 183 & -1 & 156 & 140 & 723
\end{array}\right].
\end{equation}
\end{small}
\end{figure*}

\begin{figure*}[!t]
\begin{small}
\begin{equation}\label{eqn:B4}
\setstretch{0.8}
\setlength{\arraycolsep}{1pt}
\mb B_3=\left[\begin{array}{llllllllllllllllllllllllllllllll}
-1 & 163 & -1 & -1 & 18 & -1 & 10 & 153 & -1 & -1 & 278 & 39 & 209 & -1 & 293 & -1 & 201 & 361 & 171 & -1 & 11 & -1 & 196 & -1 & 69 & 381 & 374 \\
-1 & -1 & 142 & 228 & -1 & 204 & -1 & 2 & 150 & -1 & 346 & -1 & -1 & 302 & 265 & 78 & -1 & -1 & 261 & 227 & 80 & -1 & 98 & 309 & -1 & 170 & 251 \\
-1 & -1 & 354 & -1 & 111 & -1 & 88 & -1 & 237 & 65 & 136 & -1 & 312 & 163 & -1 & -1 & 30 & 101 & 278 & -1 & 37 & 11 & -1 & -1 & 175 & 227 & 183 \\
67 & -1 & -1 & -1 & 254 & 57 & -1 & 163 & 301 & 188 & -1 & 53 & 314 & -1 & -1 & 266 & 31 & -1 & -1 & 144 & -1 & 195 & -1 & 379 & 216 & 220 & 141 \\
-1 & 348 & -1 & 347 & -1 & 279 & 307 & -1 & -1 & 13 & -1 & 250 & -1 & 249 & 176 & 274 & -1 & 129 & -1 & 54 & -1 & 224 & 68 & 205 & -1 & -1 & 20
\end{array}\right].
\end{equation}
\end{small}
\end{figure*}

We also investigate the performance of various LDPC codes via Monte-Carlo simulation.
We consider binary phase-shift keying (BPSK) transmission over the additive white Gaussian noise (AWGN) channels and the layered SPA\cite{Mansour2003cnLayer} is adopted at the receiver.
The maximum number of decoding iterations is set to 10.
 At least 100 frame errors are collected for each simulated signal-to-noise ratio point.

Fig. \ref{fig: FER} and Fig. \ref{fig: interation} show the frame error rate (FER) and the average number of decoding iterations of different QC-LDPC codes, respectively.
It can be seen that $\mb H_1$, $\mb H_2$ and $\mb H_3$ all outperform $\mb H_0 $ in terms of error correction capability and convergence speed.
Compared to $\mb H_1$, both $\mb H_2$ and $\mb H_3$ have slight performance degradation.
 This is reasonable since Strategies 2 and 3 are stricter than Strategy 1, which may eliminate some codes with good error-correction performance.

\section{Conclusion}

In this paper, we first formulated the PCM partitioning as an optimization problem for reducing the hardware complexity, which aims to minimize the maximum column  weight of each layer while maintaining a block cyclic shift property among different layers.
In particular, we  derived all the feasible solutions and proposed a tight lower bound $\omega_{LB}$ for the minimum possible maximum column weight to evaluate the quality of a solution.
Second, we reduced the computation delay of the layered decoding by considering the data dependency issue between consecutive layers.
More specifically, we illustrated how to obtain the optimal solutions with desired layer distance from those achieving the minimum value of the lower bound $\omega_{LB} = 1$.
Third, we demonstrated that up-to-now, there exist no algorithms to find an optimal solution with polynomial time complexity and alternatively proposed both greedy and enumerative partition algorithms.
Fourth, we modified the QC-PEG algorithm to directly construct PCMs that have a straightforward partition scheme to achieve $\omega_{LB} = 1$ or the desired layer distance.
Finally, we evaluated the performance of the proposed enumerative and greedy algorithms for partitioning the PCMs of the 5G LDPC codes.
For the cases where $\omega_{LB}$ or the desired layer distance is not achievable, we used the modified QC-PEG algorithm to construct two QC-LDPC codes with the same code parameters as that of a  5G LDPC code to achieve $\omega_{LB}$ and desired layer distance, respectively.
Simulation results shows that the constructed codes have better error correction performance and achieve less average number of iterations than the underlying 5G LDPC code.

\appendices


%
%

\ifCLASSOPTIONcaptionsoff
  \newpage
\fi

\bibliographystyle{IEEEtran}
\bibliography{myreference}

\end{document}